\documentclass[journal,10pt]{IEEEtran}
\usepackage{amsmath}
\usepackage{graphicx} 
\usepackage{mathtools}
\usepackage{array}
\usepackage{mathrsfs}
\usepackage{booktabs}
\usepackage{subcaption}
\usepackage{diagbox}
\usepackage{multirow}
\usepackage{amsfonts}
\usepackage{amsthm}
\usepackage{cite}
\usepackage{amssymb}
\usepackage{graphicx}
\usepackage[usenames,dvipsnames]{color}
\usepackage{epstopdf}
\usepackage[all]{xy}
\usepackage[]{mcode}
\usepackage{supertabular}
\usepackage{subcaption}
\usepackage{array}
\usepackage{hyperref}
\usepackage{algorithmic}
\usepackage{algorithm}
\usepackage{glossaries}
\makeglossaries
\newacronym{dai}{DAI}{Distributed Artificial Intelligence}
\newacronym{airdai}{AirDAI}{Distributed Artificial Intelligence over-the-air}

\providecommand{\U}[1]{\protect\rule{.1in}{.1in}}
\pdfoutput=1
\providecommand{\U}[1]{\protect\rule{.1in}{.1in}}
\newtheorem{assumption}{Assumption}
\newtheorem{theorem}{Theorem}
\newtheorem{definition}{Definition}
\newtheorem{lemma}{Lemma}
\newtheorem{appd_lemma}{Lemma}
\newtheorem{appd_theorem}{Theorem}

\begin{document}
\title{Paving the Way for Distributed Artificial Intelligence over the Air}

\author{Guoqing~Ma,~\IEEEmembership{Student~Member,~IEEE,}
        Shuping~Dang,~\IEEEmembership{Member,~IEEE,}
        Chuanting~Zhang,~\IEEEmembership{Member,~IEEE,}
        and~Basem~Shihada,~\IEEEmembership{Senior Member,~IEEE}% <-this % stops a space
\IEEEcompsocitemizethanks{\IEEEcompsocthanksitem G.~Ma,~S.~Dang,~C. Zhang~and~B.~Shihada~are~with Computer, Electrical and Mathematical Science and Engineering Division, King Abdullah University of Science and Technology (KAUST), 
Thuwal 23955-6900, Saudi Arabia  (e-mail: \{guoqing.ma, shuping.dang, chuanting.zhang, basem.shihada\}@kaust.edu.sa).}% <-this % stops an unwanted space
% \thanks{Manuscript received April 19, 2005; revised August 26, 2015.}
}

% \markboth{Journal of \LaTeX\ Class Files,~Vol.~14, No.~8, August~2015}%
% {Shell \MakeLowercase{\textit{et al.}}: Paving the Way for Distributed Deep Learning over the Air}

\IEEEtitleabstractindextext{%
\begin{abstract}
Distributed Artificial Intelligence (DAI) is regarded as one of the most promising techniques to provide intelligent services under strict privacy protection regulations for multiple clients. By applying \gls{dai}, training on raw data is carried out locally, while the trained outputs, e.g., model parameters, from multiple local clients, are sent back to a central server for aggregation. Recently, for achieving better practicality, \gls{dai} is studied in conjunction with wireless communication networks, incorporating various random effects brought by wireless channels. However, because of the complex and case-dependent nature of wireless channels, a generic simulator for applying \gls{dai} in wireless communication networks is still lacking. To accelerate the development of \gls{dai} applied in wireless communication networks, we propose a generic system design in this paper as well as an associated simulator that can be set according to wireless channels and system-level configurations. Details of the system design and analysis of the impacts of wireless environments are provided to facilitate further implementations and updates. We employ a series of experiments to verify the effectiveness and efficiency of the proposed system design and reveal its superior scalability.
\end{abstract}

}

\maketitle

\IEEEdisplaynontitleabstractindextext

% \fi
\IEEEpeerreviewmaketitle

\section{Introduction}

\IEEEPARstart{A}{s} speculated in the perspective paper `What should 6G be?' \cite{dang2020should}, sixth generation (6G) communication networks are expected to be human-centric, which poses much higher requirements for privacy protection. On the other hand, based on  existing artificial intelligence (AI) architectures, protecting digital privacy is, to some extent, contradictory to the demand of user data by intelligent communication services \cite{8677314}. This is because user data are required to be collected, processed, and utilized to precisely identify user demands so that truly intelligent and high-quality communication services can be provided to end users \cite{8928170}. These user data inevitably contain personal and sensitive information that users are not willing to share and should be restricted by legislation \cite{9060868}. Collecting and processing user data by such a centralized architecture could also lead to a high divulging risk, which becomes much more common nowadays \cite{9311931}. Moreover, relying on such a centralized architecture for intelligent communication services, one can never rule out the possibility that a malicious \textit{Big Brother} takes advantage of user data and manipulates users and even the entire society with ulterior motives \cite{6573299}.

To solve the dilemma between high-intelligence communication service and user privacy protection,
distributed deep learning (DDL) is proposed and has soon attracted researchers' attention in the communication and computing research communities \cite{8770530}. The large-scale DDL was first investigated in \cite{dean2012large} to solve the insufficient computation ability in a single node, in which a central server aggregates the one-step model gradients data updated from all agents by randomly splitting the dataset into them. However, aggregating the gradients at each stochastic gradient descent (SGD) updating increases communication overhead \cite{bottou2012stochastic}. To reduce communication overhead, local SGD has been proposed in \cite{stich2018local}, by which the multiple clients update model parameters, instead of gradients, to the central server for aggregation after a preset local SGD updating steps. Federated learning (FL) is a further advancement of local SGD \cite{konevcny2016federated}, by which only a subset of clients will update their model parameters to the central server, instead of all clients. Despite subtle differences among these learning strategies, they all belong to the family of \gls{dai} \cite{bond2014readings} due to the \textit{decoupling of client training and server aggregation}. Hence, we apply the term \gls{dai} instead of carefully distinguishing them as follows.

\begin{figure}[!t]
\centering
\includegraphics[width=3.0in]{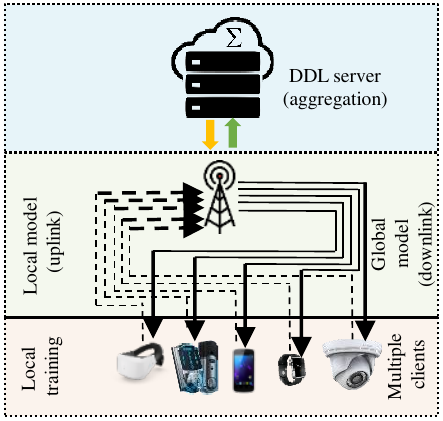}
\caption{Classic architecture of \gls{dai}, consisting of one \gls{dai} server in cloud, multiple local clients, and the uplinks/downlinks connecting them.}
\label{fldemo}
\end{figure}

Different from classical machine learning (ML) or deep learning (DL) techniques adopting centralized processing architectures \cite{9226618,9058719,9365714}, \gls{dai} utilizes a distributed processing architecture that consists of one \gls{dai} server (viz. the model owner) and multiple clients (viz. the data owners) \cite{9084352}. The clients directly collect users' raw data and process them by local training algorithms to obtain \textit{local model} parameters. 
These local model parameters are then aggregated in a certain way at the \gls{dai} server. 
The aggregated model produced at the \gls{dai} server is called the \textit{global model}, which will subsequently be updated to the clients for providing intelligent communication services. 
In this way, the global model training and first-hand raw data accessing can be decoupled, and thereby the data minimization principle for the privacy of consumer data is followed  \cite{mcmahan2017communication}. 
The classic architecture of \gls{dai} is shown in Fig.\ref{fldemo} for clarity.

Due to the distributed processing architecture and exemption from users' raw data, \gls{dai} is believed to be one of the most promising techniques to provide intelligent services under strict regulations of privacy protection \cite{8994206,9084352,9311906,9170265}. In addition, \gls{dai} can also facilitate the implementations of other promising 6G communication techniques by releasing privacy concerns and reducing the volume of data required to transmit \cite{konevcny2016federated}. Consequently, spectral efficiency, energy efficiency, and latency of communication systems would all be improved by \gls{dai} \cite{9141214}.

As described above, \gls{dai} computation is performed at both of the \gls{dai} server and clients, and the exchange of model parameters is frequent and necessary. As a result, the communication and computing procedures of \gls{dai} are coupled, which should be jointly considered and analyzed as a whole \cite{9210812}. 
Since recently, an increasing number of research works have analyzed both communication and computing issues related to \gls{dai} in wireless communication networks (details of them will be given and reviewed in the next section).
However, to the best of our knowledge, a generic system for designing and testing \gls{dai} algorithms in \textit{wireless communication networks} is still lacking, which undoubtedly impedes the development of \gls{dai} in wireless environments and \gls{dai} aided wireless networks. 
First of all, without a benchmark system, researchers interested in \gls{dai} algorithms implemented in wireless environments need to program individual communication scenarios for investigation. Also, the simulation results provided by \gls{dai} can hardly be verified by reproduction and compared with results generated by other benchmark algorithms. At last, even thought with the increasing awareness on generic design of \gls{dai} systems \cite{bonawitz2019towards, he2020fedml, li2019federated}, the researchers neglect the simulations on wireless environments, which proves to be an important factor in our work.

In this regard, we propose \gls{airdai}, a generic system design for \gls{dai} over the air, aiming at accelerating the relevant research progresses$\footnote{The  codes associated with the proposed system as well as its simulator can be found from the open GitHub repository link: \url{https://github.com/KAUST-Netlab/AirDAI}}$. 
% This generic design is supported by a system that abstracts the \gls{dai} process in realistic wireless settings. 
\begin{itemize}
    \item To ensure generality and practicability, we generalize the system design by considering a series of wireless features, including path loss, shadowing, multi-path fading, and mobility. As a result, the proposed system can be easily adapted to different settings for designing, testing, and investigating \gls{dai} applied in different wireless scenarios.
    \item We further analyze the convergence rate of \gls{dai} applied in wireless environments and affected by a set of stochastic factors.
    \item In addition, we provide a Python-written simulator according to the proposed system, and thus, it can be easily integrated into popular ML and DL frameworks, e.g., PyTorch \cite{paszke2019pytorch} and TensorFlow \cite{abadi2016tensorflow}. 
    \item Moreover, because of the generic nature, the proposed system design can be highly customized. Designers are allowed to alter the wireless communication environment and introduce self-defined quality of service (QoS) metrics with our provided simulator. 
\end{itemize}
  
The rest of the paper is organized as follows. In Section~\ref{rw}, we carry out comprehensive literature research over the works related to \gls{dai} in wireless communication networks. Summarizing the existing literature and research directions, we propose the system design in Section~\ref{system} and present the details of wireless environmental setups and convergence analysis in Section~\ref{wcm}. The effectiveness and efficiency of the proposed system design and its associated simulator are verified through several applications in Section~\ref{a}. Finally, the paper is concluded in Section~\ref{c}. For readers' convenience, we list key notations and symbols used in this paper in Table~\ref{nomenslit}.

\begin{table}[!t]
\renewcommand{\arraystretch}{1.3}
\caption{Notations and symbols used in this paper.}
\label{nomenslit}
\centering

\begin{tabular}{|c|p{4.5cm}|}
\hline
Notations and Symbols & Description\\
\hline
cell & A simulated cell containing several wireless connected clients\\
\hline
$Cr$ & Number of multiple processes used in each simulation\\
\hline
$C$ & Number of cells in each simulation\\
\hline
$M$ & Number of simulated clients per cell\\
\hline
$N$ & Number of total simulated clients in each simulation\\
\hline
$p_n$ & Ratio of local dataset size in the $n$th client\\
\hline
$r$ & Ratio of activated clients in each simulation\\
\hline
$PER$ & Packet error rate at the receiver \\
\hline
$E$ & Number of local SGD updates for each client\\
\hline
$b_s$ & Local training batch size for each client\\
\hline
$NIS_a$ & Addictive noise by naughty clients during the training phase\\
\hline
$NIS_m$ & Multiplicative noise by naughty clients during the training phase\\
\hline
$\eta_t$ & Local learning rate at the $t$th round\\
\hline

\end{tabular}
\end{table}

\section{Related Works}\label{rw}
Before planning the generic simulator design of \gls{airdai}, we need to have a profound insight into the research trends and demands of \gls{dai} in wireless communications in recent years. To well capture the research trends and demands, we carry out a comprehensive literature review over most key research works and milestones in this section. 

It has been recognized in \cite{9084352, kairouz2019advances} that communications are the critical bottleneck for \gls{dai} because of the heterogeneity of wireless networks. Therefore, communication-efficient protocols are imperative for sending messages of model updates as part of the training process, which should stipulate the number of communication rounds and the size of transmitted messages at each round \cite{han2020adaptive, wangni2017gradient, alistarh2017qsgd}.
Another core challenge mentioned in \cite{9084352} is that the unreliable connection of massive clients must be taken into consideration when modeling and analyzing \gls{dai} in wireless communication networks. Most importantly, the statistical heterogeneity of clients must be considered, which indicates that the signal propagation environments and system configurations of clients are diverse. As a result, personalized and client-specific modeling for \gls{dai} in wireless communication networks is required. 

An important application of \gls{dai} in wireless communication networks is related to mobile edge computing \cite{8770530}. In \cite{9060868}, \gls{dai} in mobile edge networks is comprehensively reviewed, and a \gls{dai} aided edge computing system is constructed. This work also summarizes three unique characteristics of \gls{dai} aided edge computing networks: Slow and unstable communications, heterogeneous clients, and privacy/security concerns. The resource allocation problems for \gls{dai} aided edge computing networks are briefly discussed, including client selection, adaptive aggregation, and incentive mechanism. It has also been pointed out in \cite{9060868} that \gls{dai} aided edge computing can help with several wireless applications, e.g., base station (BS) association and vehicular communications.

In a broader context, the motivation, opportunities, and challenges of leveraging \gls{dai} for wireless communications are discussed in \cite{9141214}. The optimization of learning time versus energy consumption by using the Pareto efficiency model and the balance between computation and communication for \gls{dai} in wireless communication networks are presented in \cite{8737464}, in which qualitative insights into \gls{dai} in wireless communication networks and a simplified multi-access communication model are provided. By the model provided, the transmission time and energy consumption for a given amount of data in \gls{dai} aided wireless communication networks are quantified. The following study on the resource allocation problems, including transmission time, energy consumption, and \gls{dai} convergence, is presented in \cite{dinh2019federated}, which has been verified to outperform the $\mathtt{FedAvg}$ algorithm by experiments on TensorFlow. A more realistic communication model for \gls{dai} in wireless communication networks is constructed in \cite{9210812}, in which learning, wireless resource allocation, and client selection are jointly optimized to minimize the \gls{dai} loss function under the constraints of latency and energy consumption. The same model is also utilized in \cite{chen2020convergence} to reduce the convergence time for \gls{dai} over wireless communication networks. 

\gls{dai} has also been utilized in more complicated wireless application scenarios, e.g., the Internet of Things (IoT), wireless sensor networks, and vehicular communication networks. In \cite{8950073}, \gls{dai} is applied to power-constrained IoT devices with slow and sporadic connections, and a fully decentralized \gls{dai} system without the \gls{dai} server is proposed. The decentralized \gls{dai} system relies on the device-to-device (D2D) communication protocols and is, in particular, suited for dense networks consisting of massive cooperative devices. In \cite{8963610}, an incentive mechanism is proposed and studied to encourage clients to contribute to \gls{dai} in the IoT. The participation of massive clients in the \gls{dai} system is formulated as a Stackelberg game, and the Nash equilibrium of the game is derived. \gls{dai} is also employed to estimate the tail distribution of vehicle's queue lengths in vehicular communication networks, which has been verified to be able to produce comparable accuracy to centralized learning methods \cite{8917592}.

\section{System Proposal}
\label{system}

\begin{figure*}[!t]
\centering
\includegraphics[width=7in]{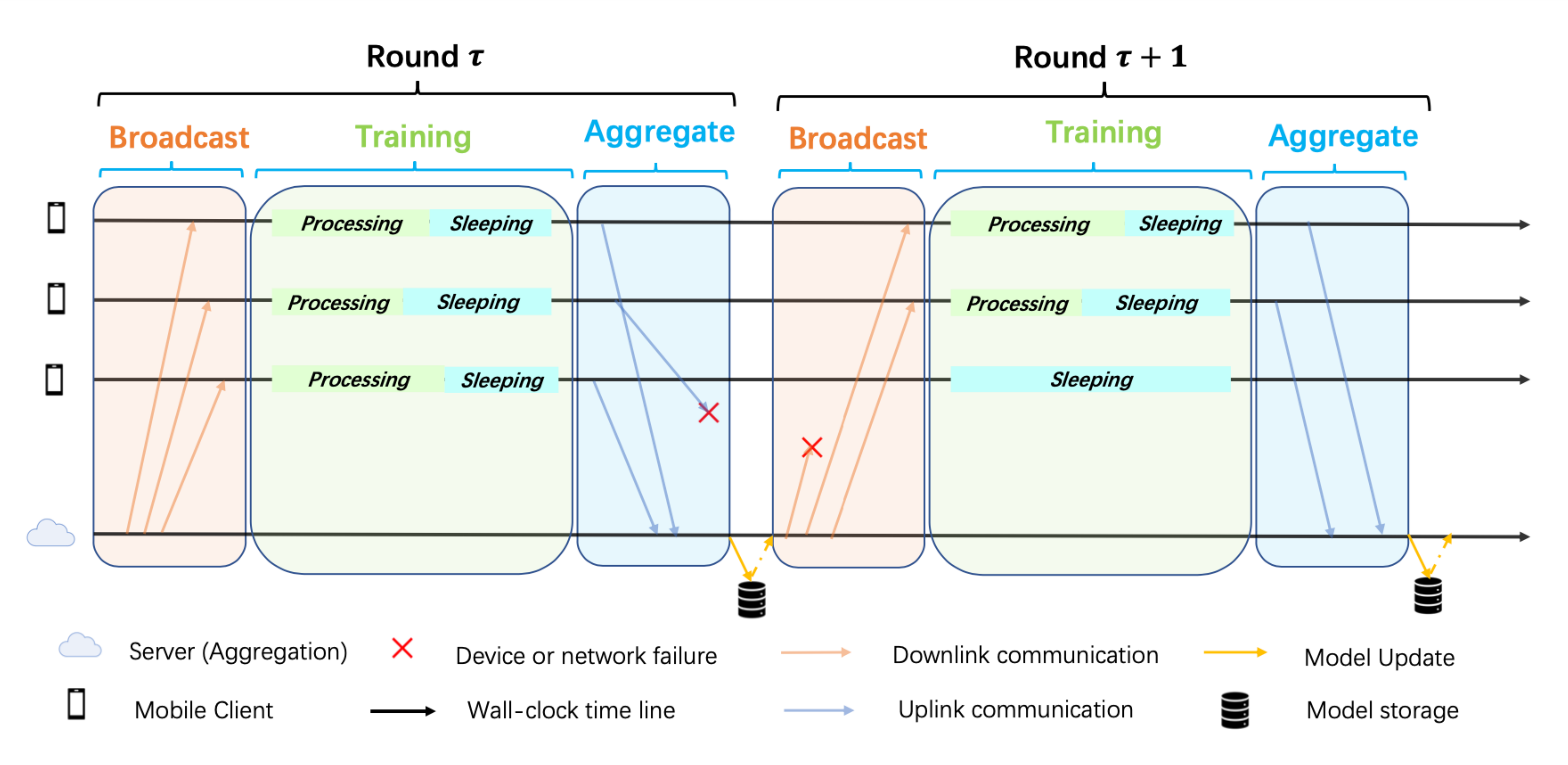}
\caption{Virtualization of one operational round of \gls{airdai}.}
\label{round}
\end{figure*}

We propose the \gls{airdai} system in this section, where we analyze and decompose the basic elements of \gls{dai} in general, introduce the programming procedures for its associated simulator, and expatiate on its scalability.
The \gls{airdai} process can be generally decomposed with two observation aspects: Spatial and temporal. 
From the spatial perspective, it can be further decomposed into the clients' sub-task and the server's sub-task for  different roles that agents play.
From the temporal perspective, the process is combined with the computing module and the communication module because of the mutually exclusive time slots on which the two modules execute. 
To make a global view on the holistic process, we virtualize and represent the \gls{airdai} tasks while concerning both aspects. 
For the convenience of illustration, we refer both clients and servers to be agents in the following without ambiguity.

\subsection{Virtualization and Basics of \gls{airdai}}
From the spatial perspective, the participants of an \gls{airdai} task in wireless environments are usually composed of several computing clients and one central server.
From the temporal perspective, \gls{airdai} tasks are processed iteratively between the clients and the server. 
Specifically, at the beginning of each time slot, clients process the pre-defined training tasks based on the local datasets and send out the computed results to the server for aggregation. 
Once received data from the clients, the server further processes data by a pre-defined aggregation function. 
Then, based on specific broadcasting strategies, the server either sends the processed data out in a limited time window or after completing the reception phase from all clients.
The interaction, which begins with the server broadcasting and ends when the server aggregates the result, is defined as a \textit{round}, as illustrated in fig \ref{round}. 

With the above explanations and settings, we represent the $\tau$th round abstractly as follows:
\begin{equation}
\label{FL_eq}
\left\{
     \begin{array}{lr}
     \mathtt{Server}:  \mathcal{K}_\tau^n \longleftarrow \mathtt{broadcast}\{\mathtt{aggregate}\{\mathcal{J}_\tau^n\}\}  \\
     \mathtt{Clients}:  \mathcal{J}_{\tau+1}^n \longleftarrow \Phi_{\mathcal{D}_n}(\mathcal{J}_\tau^n, \mathcal{K}_\tau^n)\\
     \end{array},
\right.
\end{equation}
where we utilize $\mathcal{J}_\tau^n$ to denote the message sent out from the $n$th client at round $\tau$ and $\mathcal{K}_\tau^n$ denotes message sent back to the $n$th client at round $\tau$. 
After that, it begins the ($\tau+1$)th round, and the $k$th client processes its pre-defined task $\Phi$ based on its own dataset $\mathcal{D}_n$ with received message $\mathcal{K}_\tau^n$ at round $\tau$. After finishing the computation phase, it sends the result denoted as $\mathcal{J}_{\tau+1}^n$ to the server for aggregation. 
It is worth noting that the `aggregation' and `broadcast' may only take effects on a subset of clients according to  specific policies. 
The above (\ref{FL_eq}) is a generic virtualized process covering most well-known \gls{dai} paradigms \cite{tran2019federated, bonawitz2019towards}.

\subsubsection{Synchronous and Asynchronous Settings}
Considering whether clients receive the same messages from the server during each \textit{round}, \gls{dai} schemes can be classified into synchronous and asynchronous categories \cite{sprague2018asynchronous, konevcny2016federated}. 
With the asynchronous settings, the server receives the data from a single client, then aggregates it with the historical data from other clients, and sends it back to the corresponding client before aggregating the data from newly coming clients. 
With the synchronous settings, the server has to suspend broadcasting before aggregating data from all clients or the activated clients within a pre-defined time window. The broadcast results after aggregation are identical to the activated clients during each \textit{round}. 
These schemes can be achieved by adjusting the virtual functions of $\mathtt{broadcast}\{\cdot\}$ and $\mathtt{aggregate}\{\cdot\}$  at the server end, making both the synchronous and asynchronous schemes compatible within the format of virtualization (\ref{FL_eq}).

\subsubsection{Network Topology and Virtual Channels}
To enable topological formulations, we can simply treat the agents, including both the clients and server, as vertices and the communication channels as edges. The network topology can be built as a bi-directional graph. 
Intuitively, we can represent the system as a graph $\mathcal{G}=(\mathcal{N}, \Theta)$, where $\mathcal{N}$ denotes the set of the clients and server, and $\Theta$ denotes the set of effective virtual channels.
By assigning specific parameters to corresponding vertices and edges, such as communication and computation power to different agents or Wi-Fi/LTE settings among agents, the system can be flexibly configured with varied wireless environmental settings. 

\subsubsection{QoS and Termination Conditions}
While not only paying attention to the validation accuracy or loss similar to conventional DL tasks, the proposed \gls{airdai} system focuses on the output of system QoS, e.g., total energy consumed, total time consumed, the number of activated clients per round, the number of packets lost, and etc.
Meanwhile, the server monitors the simulator states for each round and stops the simulation if one or more user-defined termination conditions are satisfied, e.g., validation accuracy reaches 98$\%$; simulation time is more than 30 minutes; total energy consumed is greater than 300 J, and etc.

\subsection{\gls{airdai} Programming Procedures}
According to the proposed system, a typical \gls{airdai} task can be generalized into three steps: building the network topology with virtual channels, defining the aggregating and broadcasting functions as well as partitioning the training dataset and building the DL model. We give introductions to all these steps as follows.

\subsubsection{Building network topology}
We provide a Python-written interface to automatically build the network topology with a specified data structure as input.
The input is organized by agents with varied attributes. 
Each agent is represented by a tree-like data structure with its identity denoting the tree root.
For each agent data structure, we arrange different layers to place the attributes according to the corresponding characteristics. 
For instance, we manually set the attribute ``role" in the first layer of each tree with different string values to distinguish between the clients and the server. 
Generally, we arrange the attributes related to the agent itself in the first layer, such as the battery capacity, the initial location and mobility speed, the computation and communication power, and etc.
Considering asymmetric channels between nodes, we cannot omit the attributes between a pair of adjacent nodes. For those attributes shared among multiple nodes, such as the virtual channels between pairs of adjacent nodes, we set the attribute ``adj" in the first layer and the adjacent node identities in the second layer with the related attributes in the third layer. Therefore, this definition of data structure is also memory efficient. An example of the data structure can be demonstrated in Fig. \ref{data_structure}.
The embedded interface will parse the data structure and complete the topology automatically.
To take advantage of the existing functions of network simulators, the underlying codes for simulating wireless networks are achieved within the system of ns-3 \cite{nsdiscrete}.

\begin{figure}[!t]
\centering
\includegraphics[width=3in]{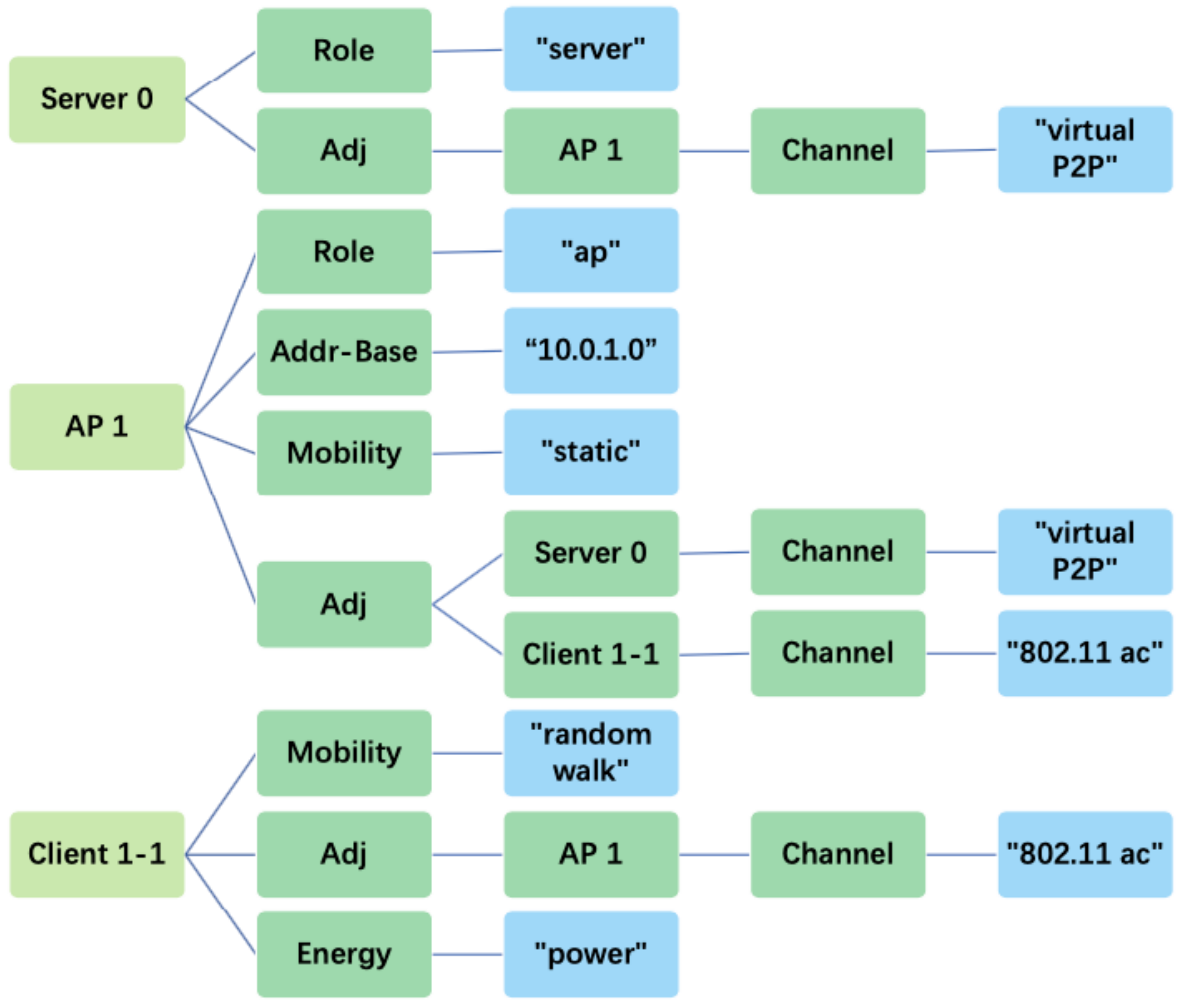}
\caption{An example of data structure for demonstration purposes.}
\label{data_structure}
\end{figure}

\subsubsection{Defining aggregation and broadcasting functions}
The system provides a programming paradigm to define personalized aggregating and broadcasting functions. 
It keeps a \textit{buffer} placeholder for each agent to receive or send new data from/to other agents and a \textit{memory} placeholder to memorize the \textit{buffer} during each round.
As a result, aggregating and broadcasting functions may only work within the activated agents in  predefined network topology during each round so as to emulate the failure of transmissions in realistic wireless environments due to certain QoS constraints.

Once an agent receives new data sequentially from the others, the \textit{buffer} will record the data and update its value according to the personalized update function. In addition, the \textit{memory} keeps tracking the latest \textit{buffer} value. Mathematically, the process can be formulated as follows:
\begin{equation}
\begin{cases}
    \mathtt{buffer} \longleftarrow \mathtt{Update}(\mathtt{buffer}, \mathtt{memory})\\
    \mathtt{memory} \longleftarrow \mathtt{buffer}
\end{cases},
\end{equation}
where the function $\mathtt{Update}(\cdot,\cdot)$ represents the user-defined buffer updating scheme. Taking the FedAvg algorithm as an example, the server averages the data currently received from activated clients during each round. Hence, $\mathtt{Update}(\cdot,\cdot)$ is the weighted average function of the latest received data and its memory. For the clients, the \textit{buffer} is the returned-back data at the end of each \textit{round}. Therefore, the \textit{buffer}  updates itself with the latest received data.

Meanwhile, the synchronous and asynchronous settings can also be achieved by determining when the server sends the updated \textit{buffer} to its adjacent client nodes. 
Specifically, when adopting asynchronous settings, the server immediately returns the updated \textit{buffer} to its recently communicating client, while with synchronous settings, the server broadcasts the recently updated buffer only after receiving data from a required number of clients.

\subsubsection{Partitioning dataset and building DL model}
The \gls{dai} tasks presume that the training dataset needs to be partitioned into multiple computing clients before training. 
We provide a paradigm to define the strategy of dataset partition. Given a predefined partition ratio of each client, each simulation process loads the identical raw dataset from shared memory and split it according to the partition ratio. Subsequently, each client with a unique rank will be assigned the corresponding sub-dataset.
If the partition ratio is not specified, the dataset will be by default partitioned into all clients in a uniform and random manner.
After dataset partition, the definition of the \gls{dai} model is just the same as the centralized counterparts.
The proposed system provides a Python wrapper function for the model to automatically aggregate and broadcast required values during the training process at each round while keeping users unaware of it unless users would like to customize their own aggregate and broadcast functions. 
For the other settings of training, users can perform exactly the same as if there were only one client in centralized tasks.

% \subsubsection{Multiple Cells}\\
% In terms of implementing the proposed FL system, there is not much difference between multiple cells and a single cell. 
% We still assume that there is a central FL server which aggregates and broadcasts data from multiple clients. However, we distinguish clients in different cells by setting a unique tag attribute for each node. Furthermore, we can define different wireless channels for multiple cells. 

% During each round of communications, when a BS is ready to aggregate data from clients in its corresponding cell, they are actually aggregated by the central FL server. 
% The central server will process the aggregated data and broadcast it back as the same as the case of a single cell.

\subsection{Scalability}
The \gls{dai} tasks in real environments usually involve a large number of computing devices with limited computing power and storage space, such as smart IoT devices and wireless sensors. 
To emulate this characteristic by limited available computing resources, e.g., a powerful workstation with several computing cores or a small computing cluster, we  implement two ingenious methods. First, the proposed system can run on multiple computing cores through distributed multi-processing interface (MPI) communication backends \cite{paszke2019pytorch}.
Before initiating simulations, the system automatically partitions the clients and the server into different computing cores and gives each core a unique rank identity. Each core maintains the identical wireless topology, in which the clients and server partition details are recorded.  
To distinguish multiple clients which are simulated in parallel but on different computing cores, we assign a unique address to each client as $(rank\_id, node\_id)$, where $node\_id$ is the index of the agent in its corresponding core.
During the processes of aggregating and broadcasting in each round, the clients send and receive data to/from the corresponding computing core where the server located through communication backends. Also, the whole communication process is unaware to users.

Second, within each computing core, we propose and utilize the scheme called ``series-tube", which provides a wrapper function and serially executes a list of objects defined on Python, to enhance the  capability of the simulator. 
Given a pre-defined object as input, the wrapper function replicates it into a list of objects according to the number of clients in a single computing core, while maintaining its functions and values as a series-tube object. 
By calling the wrapped object, the simulator serially processes the functions of the replicated objects and returns the results into a list format. 
Therefore, it keeps the whole process user-unaware and makes the codes scalable with just a few modifications.

\section{Wireless Environmental Setups and Convergence Analysis}\label{wcm}

As we introduced the system in the last section, the successful implementation of \gls{dai} in realistic transmission environments depends on the reliability of the wireless channels, over which model parameters are transmitted. It is undoubtedly that training a model in the unreliable wireless environment will degrade the efficiency compared to that in a fully reliable environment. Therefore, it is worth investigating and quantifying the impacts of the randomness of wireless channels on the training procedure of \gls{dai}. As the core of the simulator, we try to keep our design generic as much as we can and expatiate on the wireless system setups. Then, based on the given wireless environmental setups, we further analyze the convergence of a generic \gls{dai} algorithm.

\subsection{Effects of Wireless Environmental Setups}
As shown in Fig. \ref{fldemo}, there are two kinds of wireless channels pertaining to the uplink and downlink. The former refers to the transmission links from the clients to the \gls{dai} server, while the latter refers to the transmission links from the \gls{dai} server to the clients. Because the global model parameters transmitted from the \gls{dai} server are the same for all clients, we can easily adopt a broadcast protocol for the downlink transmission with sufficiently large transmit power and bandwidth, and therefore its reliability can be guaranteed. On the contrary, because all clients are required to transmit unique local model parameters, a unicast protocol is adopted for uplink transmissions. However, because clients are normally of less transmission capability, the reliability of uplink transmission is thereby problematic, and the uplink communication efficiency is of the paramount importance \cite{konevcny2016federated}. Furthermore, the unstable uplink transmission will result in a reduced number of clients' responses within a time window\footnote{The time window is dynamically managed by pace steering techniques, depending on the number of clients and service requirements \cite{bonawitz2019towards,9060868}. For example, when the number of clients is small, the time window $\epsilon$ should be set to a relatively large value so that a sufficient number of responses from clients can be collected and aggregated at the \gls{dai} server. On the other hand, when the number of clients goes large, the time window $\epsilon$ should be reduced in order to reduce the computing burden at the \gls{dai} server. The time window $\epsilon$ is, in essence, a trade-off factor between computing and communication efficiencies.} $\epsilon$, which could lead to inefficient aggregation at the \gls{dai} server and thereby a low training efficiency overall. Consequently, the wireless communication models of the uplink require special attention and are worth investigating. In the following, we specifically analyze how the randomness of wireless uplink channels affects the number of clients' responses within a predetermined time window.

Temporarily neglecting packet transmission errors, whether or not a packet from a certain client can be received is directly related to the random event that whether the transmission latency of the packet from the $n$th client, denoted as $L_n$, is less than or equal to time window $\epsilon$, $\forall~n\in\{1,2,\dots,N\}$, where $N$ is the total number of clients. Referring to the Shannon-Hartley theorem, the transmission latency $L_n$ is dominated by four factors: 1) bandwidth $B_n$; 2) transmit power $PT_n$; 3) packet size $S_n$; 4) channel power gain $G_n$. To be explicit, we can also express the transmission latency as a function of these four factors: $L_n(B_n,PT_n,S_n,G_n)$.

The first three aforementioned factors are specified by communication and \gls{dai} computing protocols and are determinate, while the last factor, i.e., the channel power gain $G_n$, is stochastic and randomly varies over time, frequency, and space. Statistically, channel power gain $G_n$ is mainly affected by four wireless propagation phenomena: 1) path loss; 2) shadowing; 3) multi-path fading; 4) molecular absorption (applicable to millimeter-wave and terahertz radios). The joint impacts of these wireless propagation phenomena can be described and simulated by different channel models, e.q., Rayleigh, Rician, and Nakagami channel models as well as a variety of compound channel models \cite{s2010wireless,rubio2007evaluation,9210812,barsocchi2006channel,proakis2008digital}, depending on the use of spectrum, node mobility, geographical and atmospheric conditions. To maintain generality, we do not specify the use of channel model in this paper.

Meanwhile, considering that there might exist errors in the received packet, error check and re-transmission are imperative in most modern communication protocols. Incorporating both mechanisms, the total transmission time of a client, denoted as $TL_n=L_n\Sigma_n$, depends on the transmission latency of a single transmission attempt $L_n$ and the number of re-transmissions $\Sigma_n$. Note that, the number of re-transmissions $\Sigma_n$ is also a random variable that is related to the coding and modulation setups and characterized by packet error rate $PER_n$. For simplicity, we can simply adopt the the geometric distribution with parameter $PER_n$ to model the random number of re-transmissions $\Sigma_n$. Based on the formulation and explanation presented above, we can simply define the packet loss rate of the $n$th client in the physical layer to be $\rho_n=\hat{F}_{TL_n}(\epsilon)=\mathbb{P}\left\lbrace TL_n>\epsilon\right\rbrace=1-F_{TL_n}(\epsilon)$, where $F_{TL_n}(\epsilon)$ and $\hat{F}_{TL_n}(\epsilon)$ are the cumulative distribution function (CDF) and the complementary CDF (CCDF) of the total transmission time $TL_n(TL_n,\Sigma_n)$ considering packet errors and re-transmissions. 

We are now able to characterize the number of clients' correct responses $\tilde{N}$ within the preset time window $\epsilon$. Assuming only the correct responses received within $\epsilon$ will be recorded at the \gls{dai} server, the number of recorded correct responses from clients $\tilde{N}$ is a dependent random number on the total transmission time $\{TL_n\}_{n=1}^N$. Because the transmissions of all $N$ clients are mutually independent, the randomness of $\tilde{N}$ can be characterized by the probability mass function (PMF) infra:
\begin{equation}\label{pmfgeneral}
\begin{split}
\Phi_{\tilde{N}}(\eta)&=\mathbb{P}\left\lbrace \tilde{N}=\eta\right\rbrace\\
&=\sum_{\tilde{\mathcal{N}}(\eta)\subseteq \mathcal{N}}\left(\prod_{n\in\tilde{\mathcal{N}}(\eta)}F_{TL_n}(\epsilon)\right)\left(\prod_{n\in\mathcal{N}\setminus\tilde{\mathcal{N}}(\eta)}\hat{F}_{TL_n}(\epsilon)\right),
\end{split}
\end{equation}
where $\mathcal{N}$ is the full set of $N$ clients and $\tilde{\mathcal{N}}(\eta)$ is an arbitrary subset of $\eta$ clients that transmit correct responses within the given time window $\epsilon$; the summation operation is carried out over all $\binom{N}{\eta}$ subsets of $\eta$ clients.

Assuming all clients are homogeneous, which implies all their channel distribution parameters and other wireless setups to be identical, we have $\rho=\rho_1=\rho_2=\dots=\rho_N$. As a result, the number of clients' correct responses $\tilde{N}$ within the preset time window $\epsilon$ abides the binomial distribution with $N$ dependent trials and  success probability $r=1-\rho$. Therefore, we can reduce (\ref{pmfgeneral}) to be $\Phi_{\tilde{N}}(\eta)=\binom{N}{\eta}r^\eta (1-r)^{N-\eta}$. When the total number of clients $N$ is large, we can rely on the law of large numbers and have the following relation:
\begin{equation}
    \tilde{N}\approx \mathbb{E}\{\tilde{N}\}=Nr.
\end{equation}
Based on this simplification, although $r$ is defined as the probability that a packet can be correctly received within the time window, it quantitatively equals the ratio of activated clients for large $N$. For notational simplicity, we denote both measures by $r$ herein, unless otherwise specified.

\subsection{Analysis of Algorithmic Convergence of \gls{dai}}
In the previous subsection, we qualitatively analyzed that the time window can influence the ratio of activated agents and thus yields an effect on the algorithmic convergence of \gls{dai}. In this subsection, we present the quantitative analysis of the convergence rate, also known as the learning rate, with respect to the ratio of activated agents.
Although the internal processes can be understood from the abstraction given in (\ref{FL_eq}), it can hardly help for analytical formulations and derivations. Hence, for facilitating the following analysis of convergence, we begin with re-defining the mathematical problem as follows:

\begin{equation}
    \min _{\mathbf{w}}\left\{F(\mathbf{w}) \triangleq \sum_{n=1}^{N} p_{n} F_{n}(\mathbf{w})\right\},
\end{equation}
where $p_{n}$ is the weight of the $n$th client such that $p_{n} \geq 0$ and $\sum_{n=1}^{N} p_{n}=1.$ Suppose that the $n$th client holds the $s_{n}$ training items: $x_{n, 1}, x_{n, 2}, \cdots, x_{n, s_{n}}$; local objective function $F_{n}(\cdot)$ is defined as

\begin{equation}
    F_{n}(\mathbf{w}) \triangleq \frac{1}{s_{n}} \sum_{j=1}^{s_{n}} \ell\left(\mathbf{w} ; x_{n, j}\right),
\end{equation}
where $\ell(\cdot ; \cdot)$ is a user-specified loss function. The problem aims at minimizing the averaged loss value through minimizing the local objective function at each distributed device. Without losing of generality, we make some common assumptions for simplifying the analysis:
\begin{itemize}
    \item $F_{n}$ is $L$-smooth function, $\forall~n\in\mathcal{N}$;
    \item $F_n$ is $\mu$-strong convex function, $\forall~n\in\mathcal{N}$;
    \item The variance of stochastic gradients in each client is bounded $\sigma^2$;
    \item The expected squared norm of stochastic gradients is uniformly bounded by $G^{2}$.
\end{itemize}
Interested readers can refer to Appendix A for mathematical implications and the inherent rationality of these assumptions.

Taking the well-known FedAvg algorithm proposed in \cite{mcmahan2017communication} as an example, we describe the process of its $\tau$th round by utilizing the abstraction given in (\ref{FL_eq}). 
Firstly, the server broadcasts the latest model parameters $\mathbf{w}_{\tau},$ to all clients, and hence, the message $\mathcal{K}_\tau^n$ received at client $n$ is $\mathbf{w}_{\tau}$ assuming a perfect downlink channel.
Secondly, every client takes the received $\mathbf{w}_{\tau}$ as the update at beginning of the local round, i.e., $\mathbf{w}_{t}^{n}=\mathbf{w}_{\tau}$, and performs $E(\geq 1)$ local SGD updates based on its own dataset:
\begin{equation}\label{fulleqwtauplus1}
\begin{array}{c}
    \mathbf{w}_{t+i+1}^{n} \longleftarrow \mathbf{w}_{t+i}^{n}-\eta_{t+i} \nabla F_{n}\left(\mathbf{w}_{t+i}^{n}, \xi_{t+i}^{n}\right),
\end{array}
\end{equation}
for $i=0,1, \cdots, E-1$, where $\eta_{t+i}$ is the learning rate, and $\xi_{t+i}^{n}$ denotes the samples uniformly chosen from the local dataset at each SGD update round. 
Thirdly, after locally updating through $E$ steps, every client sends the latest model parameters to the central server. The message sent out from $n$th client $\mathcal{J}_n^{\tau+1}$ is represented by $\mathbf{w}_{t+E}^{n}$.
Last, the central server aggregates the local models received from clients $\{\mathcal{J}_1^{\tau+1}, \cdots, \mathcal{J}_{N}^{\tau+1}\}$ to produce a new global model $\mathbf{w}_{\tau+1}$ for the next round.

Because of non-iid data distribution and partial-client participation when applying \gls{dai} in realistic wireless environments, the aggregation step can vary. Ideally, if the server receives messages from all clients (a.k.a. full-client participation) before broadcasting, the aggregation could be
\begin{equation}
    \mathbf{w}_{\tau+1} \longleftarrow \sum_{n=1}^{N} p_{n} \mathbf{w}_{t+E}^{n}.
\end{equation}
Otherwise, the partial-client participation issue rises, which can lead to low training efficiency without taking proper countermeasures. 
Specifically, the server receives the first $K$ ($1 \leq K \leq N$) messages and stops to wait for the rest. Let $\mathcal{S}_{\tau}\left(\left|\mathcal{S}_{\tau}\right|=K\right)$
be the set of the indices of the responded clients in the $\tau$th round. Then, the aggregation with partial clients' responses is performed according to
\begin{equation}\label{pareqwtauplus1}
    \mathbf{w}_{\tau+1} \longleftarrow \frac{N}{K} \sum_{n \in \mathcal{S}_{\tau}} p_{n} \mathbf{w}_{t+E}^{n}
\end{equation}
Comparing (\ref{pareqwtauplus1}) with (\ref{fulleqwtauplus1}), it is obvious that the partial-client participation issue slows down the algorithmic convergence of \gls{dai} by reducing the number of aggregated samples. The convergence of the FedAvg algorithm has been well studied when the  required number of clients is constant in \cite{stich2018local, wang2018cooperative, li2019convergence}. 
Therefore, we focus our attention on the convergence when the number of required clients is changeable among communication rounds, which reflects the realistic scenario in wireless environments, especially when we set a small time window. Our analysis is based on the recent research of federated learning on Non-IID data \cite{li2019convergence}.

Assume that the server receives $\tilde{N}_t$ (say the $t$-th communication round) activated clients within the preset time window, and assume that the total number of communication rounds is $T$.
Let $\Delta_{t}\triangleq\mathbb{E}\left\|\overline{\mathbf{w}}_{t}-\mathbf{w}^{\star}\right\|^{2}$, defined as the expected distance to the optimum, where $\overline{\mathbf{w}}_{t}=\sum_{k=1}^{N} p_{k} \mathbf{w}_{t}^{k}$ is the weighted average of model parameters among all clients, and $\mathbf{w}^{\star}$ denotes the optimized model parameters. 

\begin{lemma}
\label{lemma:one_step}
Assume that the central server received $\tilde{N}_{t}$ activated clients in the preset time window. Define $\Gamma=F^{*}-\sum_{k=1}^{N} p_{k} F_{k}^{*}$ to quantify the degree of heterogeneity of non-iid distributions.
Letting $\Delta_{t}=\mathbb{E}\left\|\overline{\mathbf{w}}_{t+1}-\mathbf{w}^{\star}\right\|^{2}$, we have
\begin{equation}
    \Delta_{t+1} \leq\left(1-\eta_{t} \mu\right) \Delta_{t}+\eta_{t}^{2} (B + C_t)
\end{equation}
where
$B=\sum_{k=1}^{N} p_{k}^{2} \sigma_{k}^{2}+6 L \Gamma+8(E-1)^{2} G^{2}$, and $C_t = \frac{N-\tilde{N}_{t}}{N-1} \frac{4}{\tilde{N}} E^{2} G^{2}$. 
\end{lemma}
\begin{proof}
Please refer Appendix A for details. 
\end{proof}

Apparently, $C_t=0$ if and only if $\tilde{N}_{t}=N$. 
Because of this inequality, we are unable to directly obtain the optimal solution. Alternatively, we can find the bound on the solution by analyzing its supremum, where we use $\sup(\Delta_{t})$ to denote the supremum of $\Delta_{t}$ for $t={1,2,\dots, T}$,  given $\eta_{t-1}$ being the learning rate at the $(t-1)$th step. Besides, we let $\sup\sup(\Delta_{t})$to denote the supremum of $\Delta_{t}$ for $t={2,3,\dots, T}$, given $\Delta_{t-1}$ reaching its supremum $\sup(\Delta_{t-1})$ at the $(t-1)$th step with $\eta_{t-2}$ being the learning rate at $(t-2)$th step. With these denotations, it follows that
\begin{equation}
\begin{cases}
    \sup(\Delta_{t+1}) = \underset{\eta_t}{\min}\left[\left(1-\eta_{t} \mu\right) \Delta_{t}+\eta_{t}^{2} (B + C_t)\right]\\
    \sup\sup(\Delta_{t+1}) = \underset{\eta_t}{\min}\left[\left(1-\eta_{t}     \mu\right) \sup(\Delta_{t})+\eta_{t}^{2} (B + C_t)\right]
\end{cases}
\end{equation}
$\forall~t=1, 2,\dots, T-1$, by which we can determine the minimum by
\begin{equation}
\begin{cases}
    \sup(\Delta_{t+1}) = \Delta_{t} - \frac{\mu^2\Delta_{t}^2}{4(B + C_t)}\\
    \sup\sup(\Delta_{t+1}) = \sup(\Delta_{t})- \frac{\mu^2\sup(\Delta_{t})^2}{4(B + C_t)} .
\end{cases}
\end{equation}
For the quadratic function $f(x) = x-\frac{\mu^2x^2}{4(B+C)}$, we can obtain its maximum to be $\frac{B+C}{\mu^2}$ when $x=\frac{2(B+C)}{\mu^2}$ and derive $f(x_1) \leq f(x_2)$ when $x_1 \leq x_2 \leq \frac{2(B+C)}{\mu^2}$. As a result, letting $x=\Delta_{t-1}$, we know that $\sup(\Delta_t) \leq \frac{B+C_{t-1}}{\mu^2}$. Because of $B > C_t$, $\forall~t={1,2,\dots, T}$, we can derive the inequality $ \frac{B+C_{t-1}}{\mu^2} \leq \frac{2(B+C_t)}{\mu^2}$. Finally, we have $\sup(\Delta_{t+1}) \leq \sup\sup(\Delta_{t+1})$.

Recursively let 
\begin{equation}
    \Tilde{\Delta}_{t+1} = \underset{\eta_t}{\min}\left[\left(1-\eta_{t}\mu\right) \Tilde{\Delta}_{t}+\eta_{t}^{2} (B + C_t)\right],
\end{equation}
for $t = {0,1,\dots, T-1}$, and let $\Tilde{\Delta}_{0} = \Delta_{0}$. Given $t'<t$, it can be found that $\Tilde{\Delta}_{t}$ is the supremum of $\Delta_t$ by setting all its previous $\Delta_{t'}$ being the corresponding supremum. 
With the analysis above, we know that the supremum converges fastest when $\eta_t = \frac{\mu\Tilde{\Delta}_{t}}{2(B+C_t)}$. 
With the above analysis, we want to find the relations between the learning rates of partial device participation and full device participation conditions. The result is presented as follows.

\begin{lemma}
\label{lemma:linear_ratio}
Denote $\bar{\eta}_t$ to be the learning rate at communication round $t$ to guarantee the algorithm convergence when full devices are participated.
Let $r_t = \frac{\tilde{N}_t}{N}$ be the device participation ratio at communication round $t$. The convergence of the algorithm when partial devices are participated can be guaranteed by setting $\eta_t = r_t\bar{\eta}_t$.
\end{lemma}

\begin{proof}
Hint: By analyzing the relation of learning rates between $C_t=0$ and $C_t>0$, we can find an equation to combine the two conditions. Please refer Appendix A for details. 
\end{proof}

Withe analysis of Lemma \ref{lemma:one_step} and \ref{lemma:linear_ratio}, we can begin to analysis the convergence rate in the wireless environments as follows.

\begin{theorem}
Let the assumptions hold and $L, \mu, \sigma_{k}, G$ be defined therein. Choose $\kappa=\frac{L}{\mu}$, $\gamma=\max \{8 \kappa, E\}$ and the learning rate $\eta_{t}=\frac{2r_t}{\mu(\gamma+t)} .$ Then FedAvg algorithm in wireless environments satisfies
$$
\mathbb{E}\left[F\left(\mathbf{w}_{T}\right)\right]-F^{*} \leq \frac{2 \kappa}{\gamma+T}\left(\frac{B+D}{\mu}+2 L\left\|\mathbf{w}_{0}-\mathbf{w}^{*}\right\|^{2}\right),
$$
where $B=\sum_{k=1}^{N} p_{k}^{2} \sigma_{k}^{2}+6 L \Gamma+8(E-1)^{2} G^{2}$, and $D = 4E^2G^2$.
\end{theorem}

\begin{proof}
Hint: Assume $C_t=0$ and from Lemma \ref{lemma:one_step}, find the bound of $\Delta_t$ by induction. Apply the assumptions on $F$, find the relations between $F(\mathbf{w}_t)$ and $\Delta_t$. Combining with Lemma \ref{lemma:linear_ratio} to find the learning rate in wireless environments. Please refer Appendix A for details.
\end{proof}

\section{Experiments}\label{a}

In this section, we take the well-known FedAvg algorithm as an example to validate the effectiveness of the proposed system. In particular, we systematically evaluate the performance of FedAvg with different parameter settings, while the parameters can be roughly split into two categories: model-related hyper-parameters and system-related parameters. The target of a series of experiments is to study the accuracy, efficiency, robustness, and fairness of a given algorithm based on our proposed system. Besides, we also validate the scalability of the system.

\subsection{Experiment Setup}
To demonstrate the generality of our proposed system, we consider two completely different tasks on the PyTorch platform. 
The first one is a multi-class image classification problem for digital recognition, and the second one is a regression problem for wireless traffic prediction \cite{zhangcl2018, zhangjsac2019, zhanginfocom2021}.
We perform the first task on the MNIST dataset \cite{Lecun_mnist}. 
This dataset is one of the most classical ones in the ML/DL realm and has been widely applied in the literature.
For the multi-class classification problem, we attempt to predict which class the input image belongs to, and the prediction accuracy is adopted as the evaluation metric. 
%As one of the most classical problems in the ML realm, multi-class image classification on the MNIST dataset has been the most widely used.
In the experiment, the model architecture adopted for this task is described as follows: A CNN with two $5\times 5$ convolution layers (the first layer with 10 channels, and the second layer with 20 channels; each followed by a $2\times 2$ max pooling and the rectified linear unit (ReLU) activation function), a fully connected layer with 50 units utilizing the ReLU activation function for neural computing, and a final softmax output layer \cite{mcmahan2017communication}. The total number of the applied model parameters equals 1199882.
The initial learning rate is set to unity, with an exponential decay rate at 0.9 every 5 local training steps.

We perform the second task on the Call Detail Record (CDR) dataset from `Telecom Italy Open Big Data Challenge' \cite{barlacchi2015multi}.
The CDR dataset contains three kinds of wireless traffics from different cells: The number of text messages, the number of calls, and the number of Internet data packages.
For this problem, we attempt to predict the future traffic volume of a cell, given the historical traffic volumes, and the mean square error (MSE) is adopted as the evaluation metric. 
In the experiment, the model architecture adopted includes a stacked long short-term memory (LSTM) structure with two LSTM layers (each layer with 64 hidden units) and a fully connected layer with a single output. The total number of the applied model parameters equals 12961.
The initial learning rate is set to be\footnote{Note that the models and (hyper-)parameters we adopted here are rather straightforward since the design and optimization of network architecture and (hyper-)parameters are out of the scope of this paper.} 0.05, with an exponential decay rate at 0.9 every 5 local epochs.

We assume that all the clients connect to the server through wireless links.
In the following experiments, if without further annotations, we assume all computing clients are located randomly in several wireless cells. Each cell is simulated within one CPU core process, while the server is simulated in another independent core process.
Within each cell, we assume that there are a limited number of computing clients randomly walking in a squared area and communicating with the server through an AP node. We further assume the wireless channel model for each client to be the constant speed propagation delay model and log distance propagation loss model. 
Without losing generality, we assume that each AP node connects to the server through a virtual point-to-point link with a limited data rate and delay. 
For the sake of simplicity, we assume that all clients in every cell have the same system configurations and adopt the suggested channel and system parameters in \cite{carroll2010analysis}, which are listed in Table \ref{nomensvalue}.

The computing time is closely related to the CPU frequency, IO throughputs, memory cache, and the existing tasks running on the agent's device and thereby hard to mathematically formulate.
To simulate the computing time of agents precisely, we assume it to be 10 times the computing time on our computational platform, which is a workstation with two physical CPUs, 20 core processes per CPU, and 256 GB memory cache.
To avoid interference from the existing tasks running on the workstation, we simultaneously build the simulations for each experiment to keep the same operational conditions.
These system configurations are fixed unless otherwise specified.
% Specifically, the path loss $\alpha_n$, reference distance $d_0$, cell service radius $\delta$, fading figure $m_n$, and shadowing figure $\sigma_n$ are set to be 2.5, 3.5, 50, 1.0 and 2.0, respectively. 
% The transmit power of each client, $PT_n$, is $10$ mw, and the transmit power of server node, say $PTS$, is $150$ mw.

\begin{table}[!t]
\renewcommand{\arraystretch}{1.3}
\caption{Suggested values and ranges as well as explanations of independent channel distribution and system configuration parameters.}
\label{nomensvalue}
\centering

\begin{tabular}{|p{1.2cm}|p{6.5cm}|}
\hline
\textit{Parameters} & \textit{Suggested values, ranges, and explanations}\\
\hline
$B_n$ & Bandwidth of wireless communications, 10-29 Mbps for 802.11g, 150 Mbps for 802.11n, 3-32 Mbps for 802.11a, 210 Mbps - 1 Gbps for 802.11ac; We choose 10 Mbps here.\\
\hline
$PT_n$ & Wireless transmission power, 80-720 mW for WiFi modulem; we choose 720 mW here.\\
\hline
$S_n$ & Communication packet size, usually less than 64K bits for UDP and TCP protocols; We choose 1K bits here.\\
\hline
$G_n$ & 1.4125-2.2387 for LoS links, 2.1878-3.0549 for non-LoS links; We choose 1 here.\\
\hline
$M$ & 1-6 clients per cell; We choose 4 clients per cell here.\\
\hline
$datarate$ & Data rate of P2P channel between server and Ap node; We choose 500 Mbps here.\\
\hline
$delay$ & Delay of P2P channel between server and Ap node; We choose 20 ms here.\\
\hline
\end{tabular}
\end{table}

\subsection{Accuracy}
%***********************************Figure BEGIN***************************************
\begin{figure*}[!t]
    \centering
    \begin{subfigure}[t]{0.5\textwidth}\label{mnist_acc}
        \centering
        \includegraphics[width=3.4in]{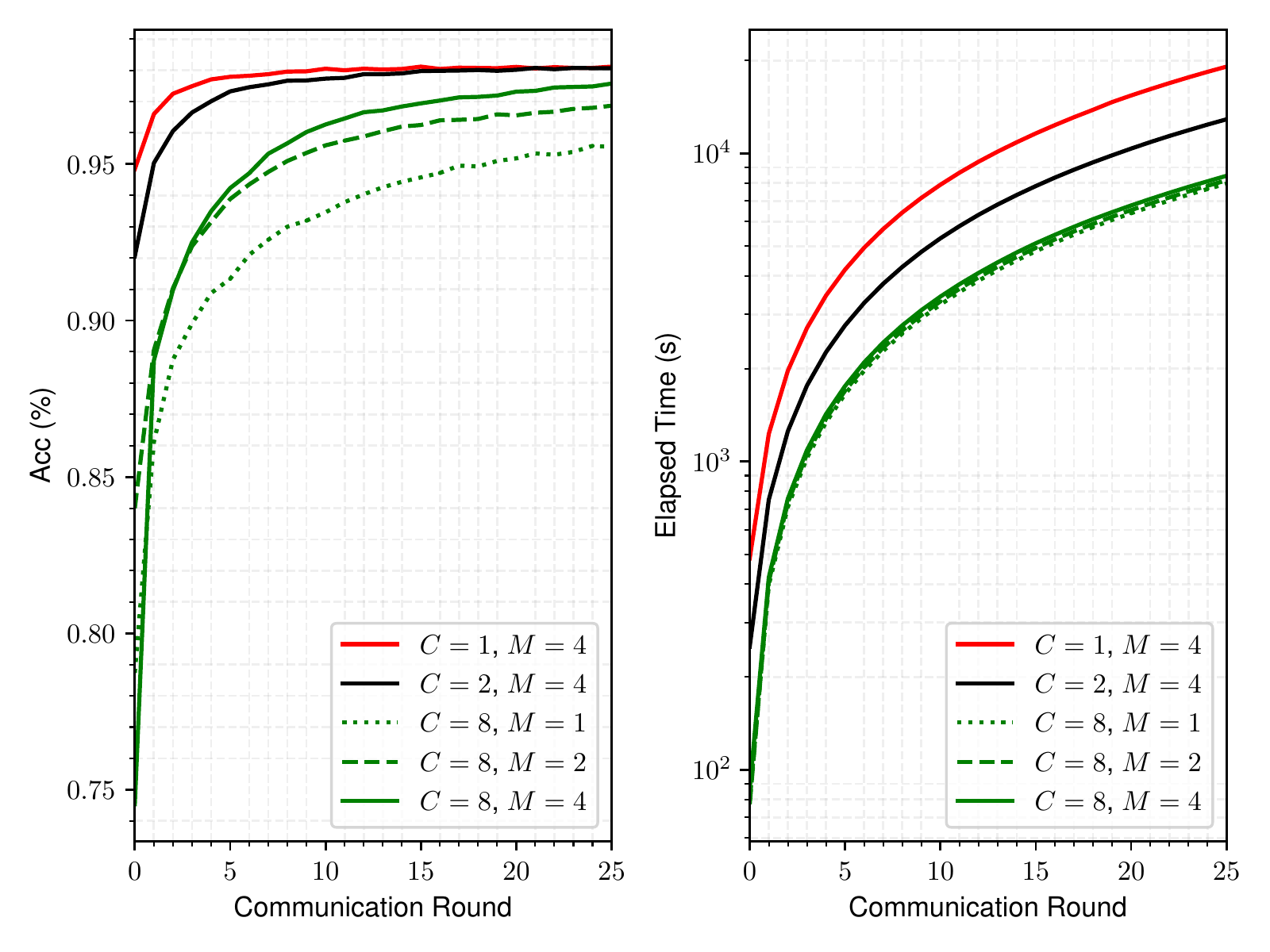}
        \caption{Image classification task.}
    \end{subfigure}%
~
    \begin{subfigure}[t]{0.5\textwidth}\label{traffic_acc}
        \centering
        \includegraphics[width=3.4in]{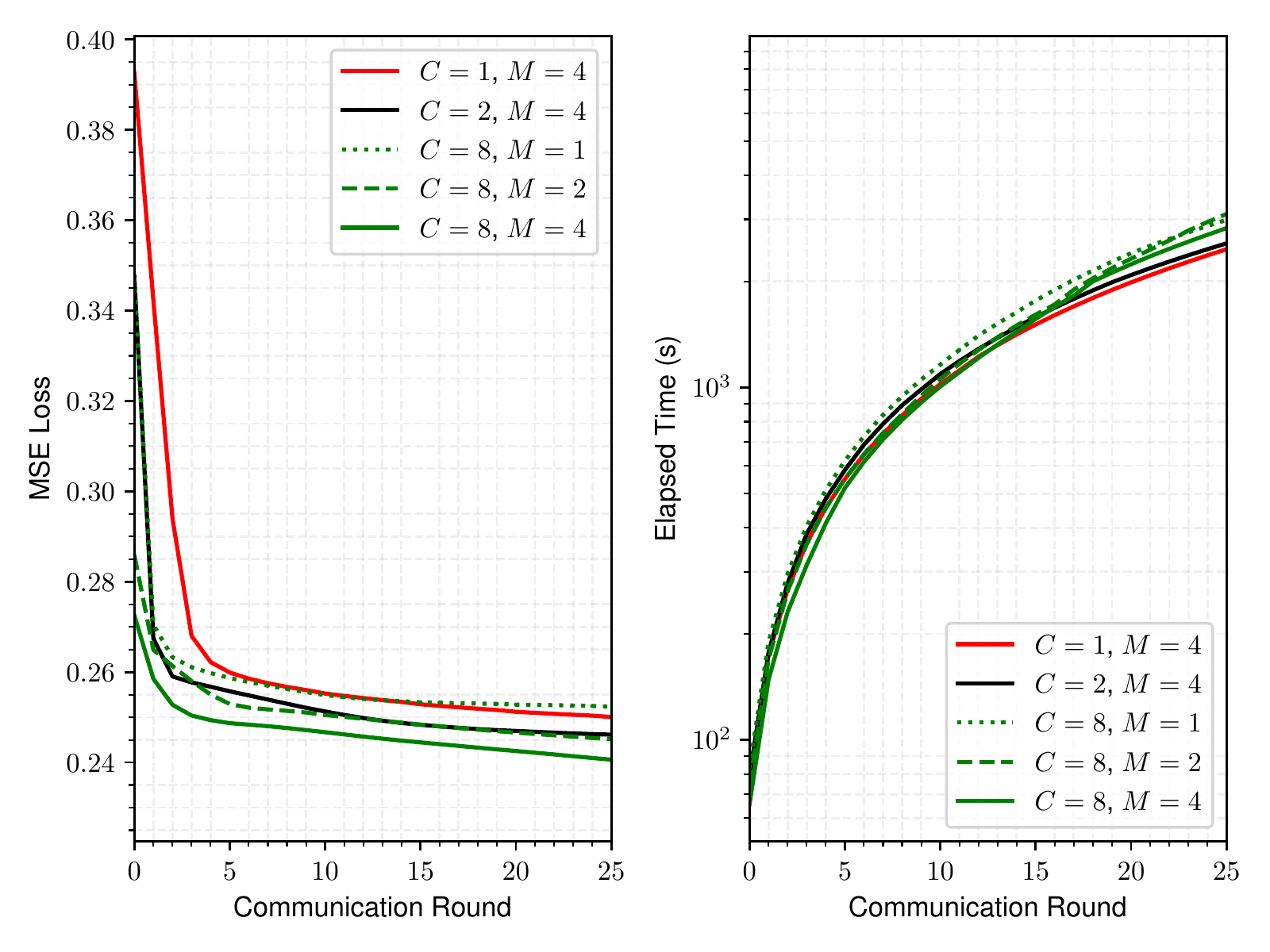}
        \caption{Traffic prediction task.}
    \end{subfigure}%
    
    \caption{Accuracy and time of two real-world ML tasks versus epochs, given different number of WiFi cells and different number of activated clients in each cell.}
    \label{performance_comparisons}
\end{figure*}
%***********************************Figure END*****************************************

In this subsection, we present the overall prediction performance of our simulator. The experiments are conducted as follows. We set the number of cells to 1, 2, and 8, respectively. Each cell is simulated in a single CPU core. Within each cell, we assume there are 4 activated agents. 
Besides, we also consider that the number of active agents per cell to be 1 and 2 when the number is 8. 
Thus, we have 5 scenarios in total. 
For each scenario, we assume that each agent has a sub-dataset with the same size and distribution.
Furthermore, we assume that the image classification task in different experiments has the same size as the whole dataset. However, we assume that the sub-dataset size for the traffic prediction task is constant, implying that the whole dataset size increases with the number of clients. We also stipulate different learning rates according to Lemma \ref{lemma:linear_ratio} for different activated ratio scenarios.
Specifically, we set the learning rate ratio to be the sub-dataset size dividing the whole dataset size for each client.

%***********************************Figure BEGIN***************************************
\begin{figure*}[!t]
    \centering
    \begin{subfigure}[t]{0.5\textwidth}\label{mnist_loss}
        \centering
        \includegraphics[width=3.4in]{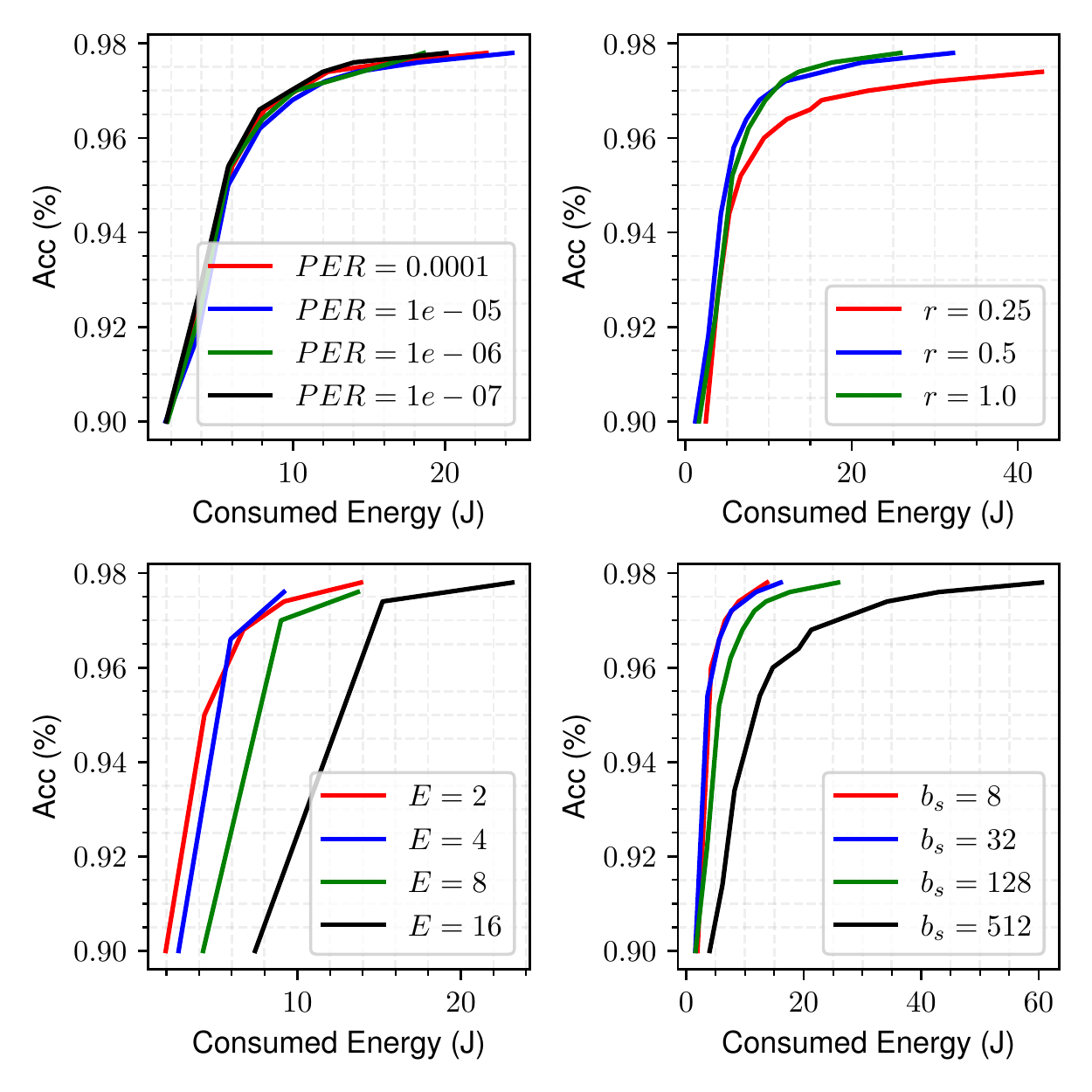}
        \caption{Image classification task.}
    \end{subfigure}%
~
    \begin{subfigure}[t]{0.5\textwidth}\label{traffic_loss}
        \centering
        \includegraphics[width=3.4in]{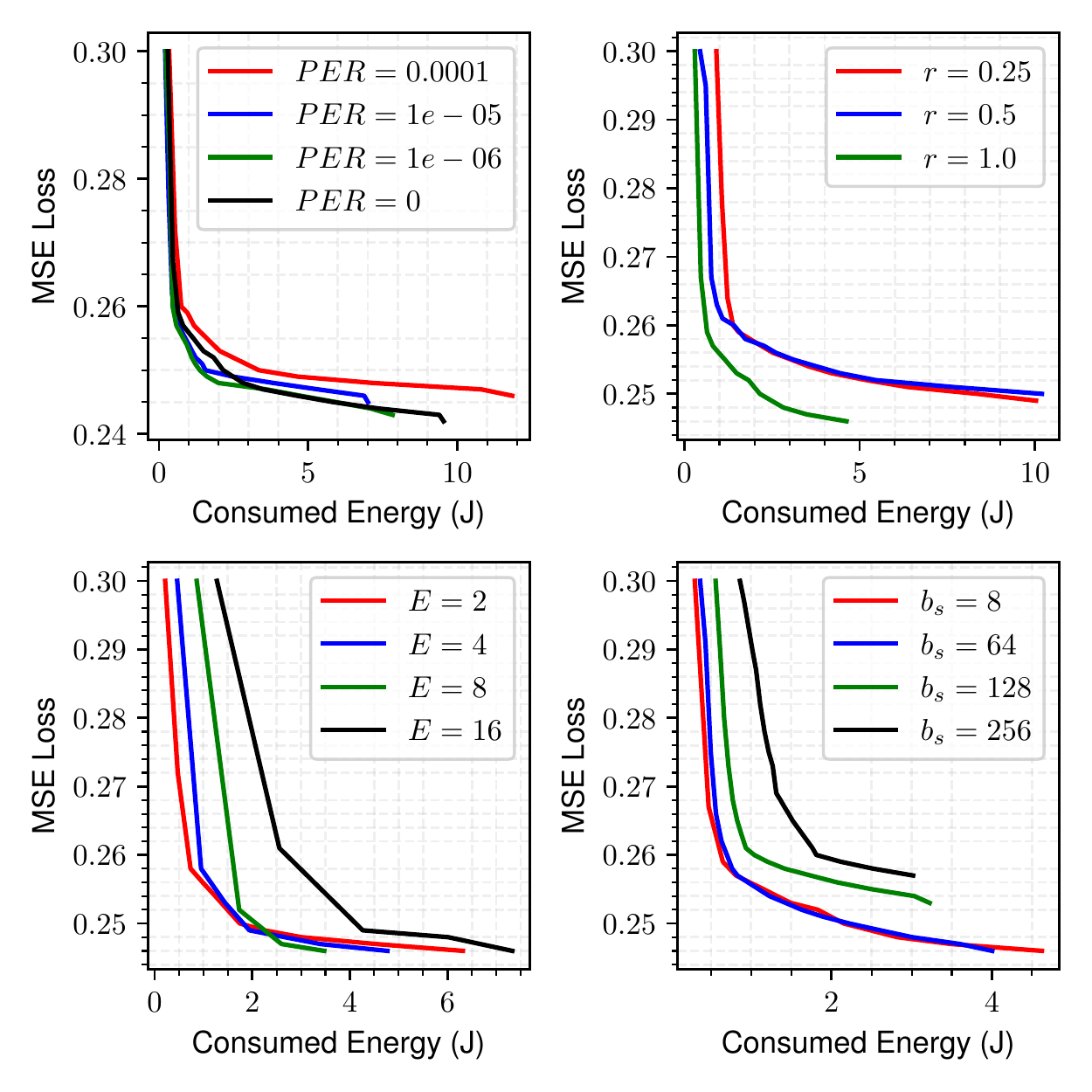}
        \caption{Wireless traffic prediction task.}
    \end{subfigure}
    \caption{Test dataset performance of two real-world ML tasks versus energy with four variables: Error rate, active ratio, local epochs, and batch size.}
    \label{acc_energy}
\end{figure*}
%***********************************Figure END*****************************************

%\subsubsection{Performance of Classification and Regression Tasks}

We utilize the accuracy and MSE loss on an independent test dataset to represent the performance of both image classification and wireless traffic prediction tasks, respectively.
To clearly present results the figures, we utilize different colors of red, black, and green to denote the cases corresponding to the number of cells of 1, 2 and 8, respectively.
Besides, we utilize dot-line, dash-line, and solid-line to represent the cases with the number of active agents per cell of 1, 2, and 4, respectively.
As the results presented in both subfigures of Fig.~\ref{performance_comparisons}, we draw two sets of lines to represent the results of performance and time versus the number of training rounds. 

From Fig.~\ref{performance_comparisons}, it is clear to observe that the scenario with 1 cell and 4 active agents achieves the best performance among all scenarios in the image classification task, while in the wireless traffic prediction task, the scenario with 8 cells and 4 active agents per cell outperforms. However, for the scenarios with the same number of cells in both tasks, the more active agents per cell will lead to better performance.
These phenomenons can be applied to explain both the weakness and strength of the FedAvg algorithm. When the number of active clients per cell equals 4, no matter how many cells are utilized for training, the whole training dataset keeps unchanged for the image classification task. The mathematical theory has proved the convergence of FedAvg. However, it does not perform as well as a centralized algorithm in practice. At least, it converges slower than a centralized algorithm. 
Nevertheless, the conclusion is just the opposite in the wireless traffic prediction task. The dataset is not static, and the more cells are utilized in the training phase, the larger the training dataset is. A larger training dataset generally yields better prediction performance. The FedAvg algorithm, as a result of this, works better with a large number of cells, which reflects the negative influence caused by increasing the number of cells.

As for the consumed time, the two tasks perform differently as usual. 
For the image classification task, the scenarios with a small number of cells spend a large amount of time to finish the same number of rounds. In contrast, for the traffic prediction task, the conclusion is the opposite. 
The reason is that the computing phase takes a dominant position compared with the communication phase in the image classification task. Therefore, the scenarios with a small number of cells spend more time on the computation part than the scenarios with many cells.
In the traffic prediction task, the computational time consumed is almost the same for all cells, as they have the same size of sub-dataset for training. Therefore, the scenarios with a large number of cells need more time for communications than the scenarios with a small number of cells, which causes the opposite results to the image classification task.

\subsection{Efficiency}
In this subsection, we study the factors that affect the efficiency of the FedAvg algorithm.
We define the efficiency of our system as the energy and time consumed for a task to reach the termination condition. Specifically, we set the termination condition for our experiments when the FedAvg algorithm reaches an accuracy or a loss threshold.
The studied variables include the receiver error rate, the agent activated ratio, the number of local training epochs, and the training batch size. Some other variables may also affect the efficiency, but we only study variables as mentioned above due to their dominant and direct impacts.

In particular, we set the accuracy thresholds for the image classification task starting from 0.9 and ending at 0.98 with 0.002 as steps and set the loss thresholds for the traffic prediction task starting from 0.3 and ending at 0.245 with -0.001 as steps.
We accumulate the consumed energy and time for each threshold to draw simulation curves. From previous experience, the simulation curves would present a ladder shape if the energy has not changed between two consecutive thresholds. Therefore, we only keep the first result if the energy value is constant among several consecutive thresholds. The performance on energy and time with different settings are presented in Fig. \ref{acc_energy} and Fig. \ref{acc_time}. We discuss the simulation results for two tasks affected by different settings separately as follows.

%***********************************Figure BEGIN***************************************
\begin{figure*}[!t]
    \centering
    \begin{subfigure}[t]{0.5\textwidth}\label{mnist_loss}
        \centering
        \includegraphics[width=3.4in]{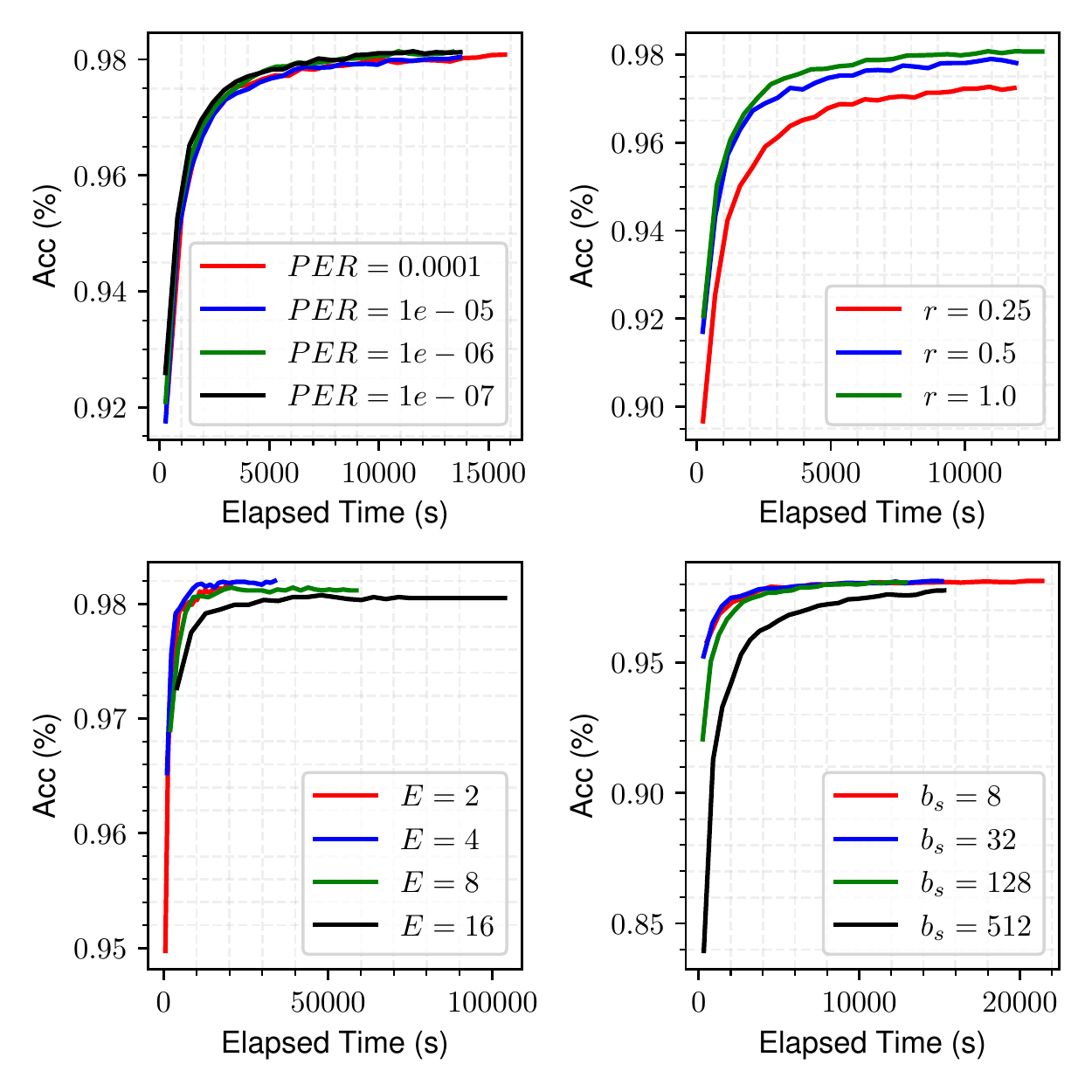}
        \caption{Image classification task.}
    \end{subfigure}%
~
    \begin{subfigure}[t]{0.5\textwidth}\label{traffic_loss}
        \centering
        \includegraphics[width=3.4in]{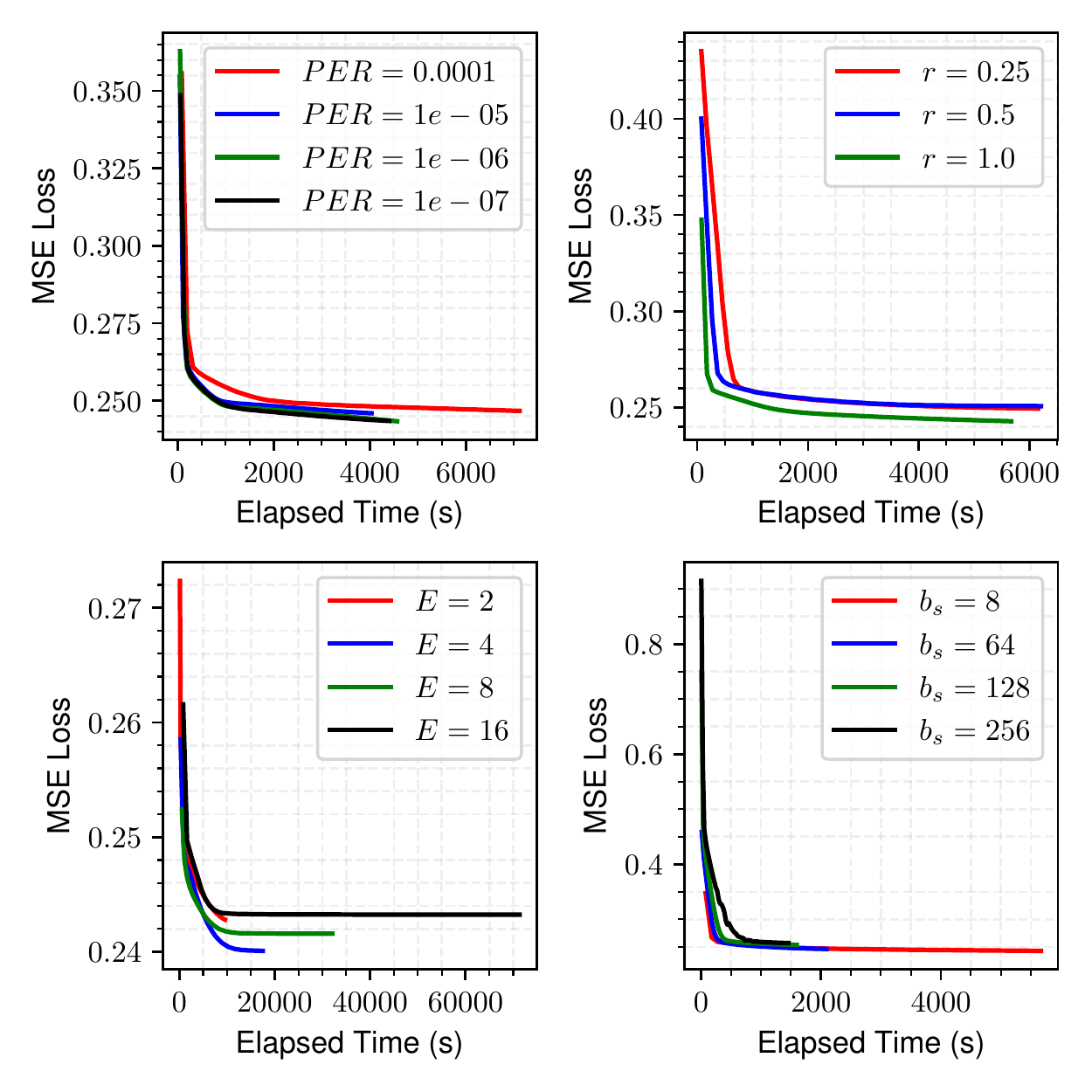}
        \caption{Wireless traffic prediction task.}
    \end{subfigure}
    \caption{Test dataset performance of two real-world ML tasks versus time with four variables: Error rate, active ratio, local epochs, and batch size.}
    \label{acc_time}
\end{figure*}
%***********************************Figure END*****************************************

%\subsubsection{Effect on Energy and Time Consumed}

The packet error rate is chosen from $[1e-4, 1e-5, 1e-6, 1e-7]$. The figures show that varying packet error rate does not affect the test dataset convergence, as they reach the same maximum accuracy or minimum MSE loss. However, it makes sense and can also be observed that a significant error rate will considerably increase the energy and time to reach convergence.

As represented in the figures, a large active ratio performs better than a small active ratio. The reason is that compared to a small active ratio, a larger one has more datasets involved in the training phase, which makes the test performance reach the same value while consuming less energy. As for the time consumed, the conclusion is not so clear. A larger training dataset generally converges faster than a smaller one. However, a large active ratio may increase the time for communications, resulting in an increase in the total time consumed. Although the figures in our experiments present that a larger active ratio consumes less energy and time to reach the same test performance, we cannot conclude that a large active ratio will always be helpful.

The number of local epochs refers to the number of training epochs for each client during the training phase. A generally accepted common knowledge is that increasing the number of training epochs will significantly decrease the communication over computation ratio and require fewer communication rounds to complete the same total number of epochs. This will lead to faster convergence than that with a small number of local epochs. However, the results present a counter-intuitive conclusion. There might be two reasons for this phenomenon. First, the computation time takes a significant ratio of a complete round compared to the communication time. Second, it depends on the algorithm, and when the number of local epochs reaches a threshold, further increasing it will not accelerate the convergence of corresponding tasks.

The simulation results also show that the batch size only affects the training phase. The optimal batch size to reduce energy and time cost for one round depends on the specific tasks and the computation power of the agents. In our experiments, as shown in the figures, 32 is the best choice for the image classification task among all other options, while 64 is the best for the traffic prediction task.

\subsection{Robustness}
Any practically implementable algorithm must be robust to malicious users in reality \cite{bhagoji2019analyzing, bagdasaryan2020backdoor}. Based on our system, we carry out experiments on the FedAvg algorithm to validate its robustness to malicious agents. We assume that the agents are malicious and spam erroneous data to the central server, while the erroneous data in the following simulations is produced by adding a Gaussian noise to the original data. It is worth noting that the added noise strength must less than a threshold, otherwise, the central server can easily distinguish the malicious agent by comparing it with the average value and will reject the malicious data. We set up the experiments by considering two kinds of noise: Additive noise and multiplicative noise. The additive noise is generated as $w_{\mathrm{noise:a}} = w+\mathcal{N}(0, NIS_\mathrm{a})$, and the multiplicative noise is generated as $w_{\mathrm{noise:m}} = w\times(1+\mathcal{N}(0, NIS_\mathrm{m}))$, where $w$ is a model parameter capturing the baseline of the correct data, and $\mathcal{N}(0, NIS)$ is a zero-mean and $NIS_{a/m}$-deviation Gaussian distributed random variable.

%\subsubsection{Effect of Naughty Agents}
%***********************************Figure BEGIN***************************************
\begin{figure*}[!t]
    \centering
    \begin{subfigure}[t]{0.5\textwidth}\label{mnist_loss}
        \centering
        \includegraphics[width=3.4in]{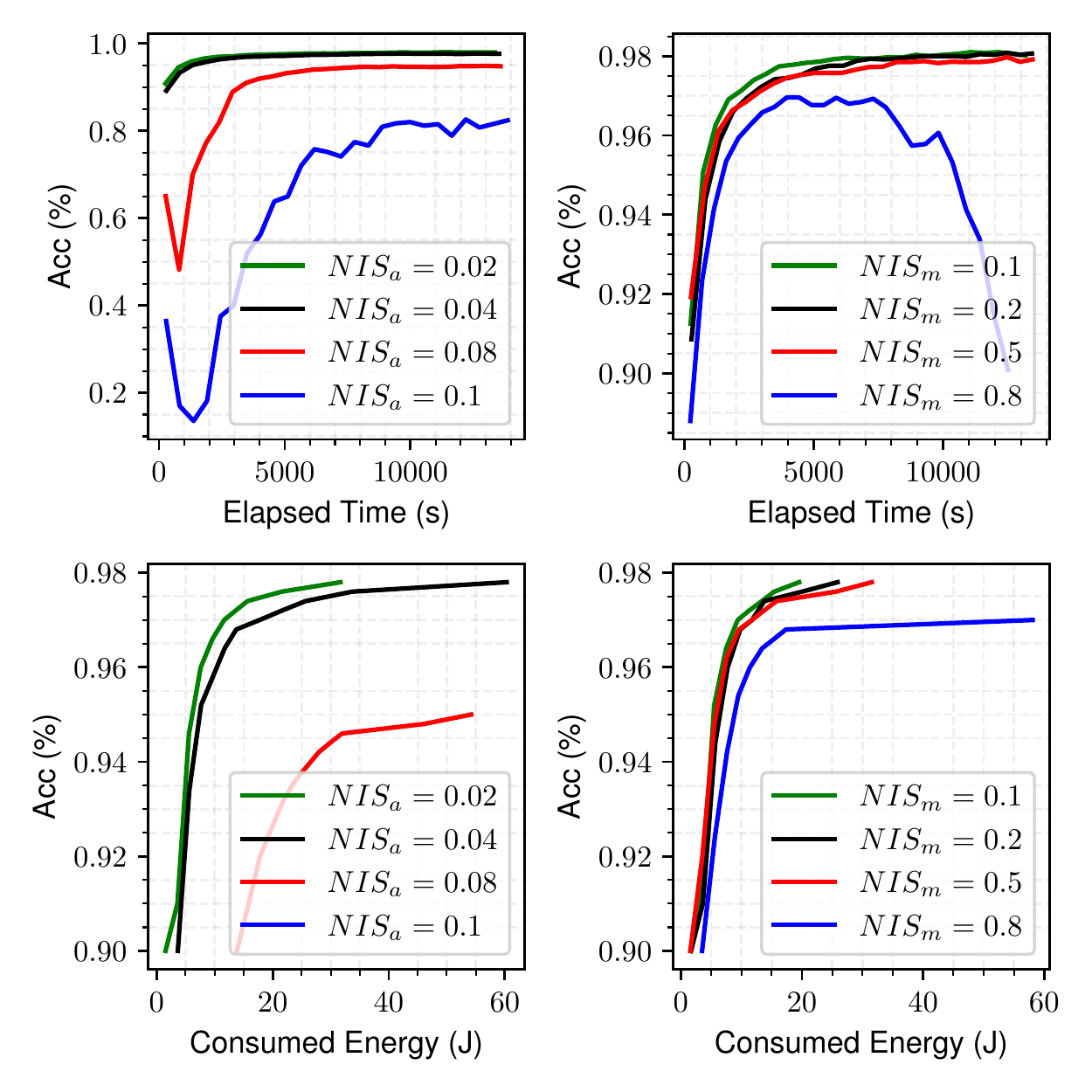}
        \caption{Image classification task.}
    \end{subfigure}%
~
    \begin{subfigure}[t]{0.5\textwidth}\label{traffic_loss}
        \centering
        \includegraphics[width=3.4in]{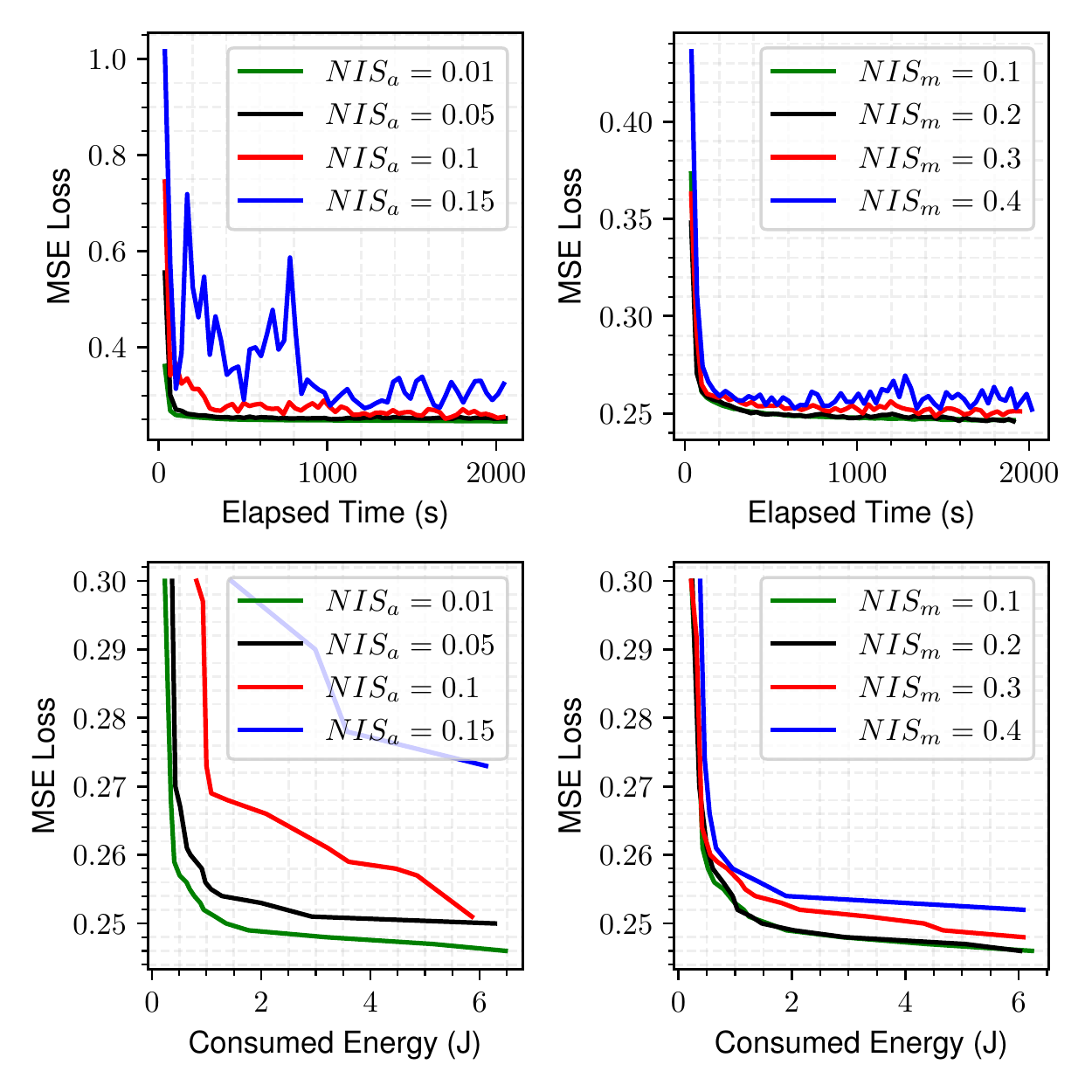}
        \caption{Wireless traffic prediction task.}
    \end{subfigure}
    \caption{Training performance comparisons among different types of noise on the real-world two ML tasks.}
    \label{fig_naughty}
\end{figure*}
%***********************************Figure END*****************************************

The simulation results regarding the robustness test are shown in Fig.\ref{fig_naughty}, from which one straightforward observation is that the same noise will have different effects on different tasks. For instance, the performance has been significantly degraded for the classification task when $NIS_{a}$ of the additive noise equals 0.1. In contrast, the traffic prediction still has a competitive performance with the same additive noise.
We can observe a similar phenomenon when applying the multiplicative noise. The classification task has a higher level of robustness to the multiplicative noise than the prediction task. Moreover, even though the slight value noise has an in-distinctive impact on the accuracy or MSE loss performance, to reach the same performance, it will consume more energy and time compared to the benchmark without noise. In summary, even applying the same FedAvg algorithm under the same experimental conditions, different tasks with different model parameters will vary from different noise levels.

\subsection{Fairness}
% \begin{table}[!t]
% \centering
% \caption{ Energy and time consumed with different dataset partitions}\label{tab:fairness}
% \renewcommand\arraystretch{2}
% \begin{tabular}{|c|c|c|c|c|}
% \hline
% \multicolumn{5}{|c|}{\textbf{Image Classification Task}} \\\hline
% \textit{partition} & \textit{energy}  & \textit{time}   &  \textit{acc}  &  \textit{ratio} \\\hline
% 8,1,1,1	 & 1.030	 & 2.287	 & 0.980	 & 6.190	 \\\hline
% 64,1,1,1	 & 1.059	 & 2.950	 & 0.980	 & 19.462	 \\\hline
% 512,1,1,1	 & 1.058	 & 3.023	 & 0.980	 & 25.548	 \\\hline
% 4096,1,1,1	 & 1.074	 & 3.064	 & 0.980	 & 25.259	 \\\hline
% 512,512,512,1	 & 0.996	 & 1.233	 & 0.980	 & 14.031	 \\\hline
% 4096,4096,4096,1	 & 0.994	 & 1.235	 & 0.980	 & 14.021	 \\\hline
% \end{tabular}
% \\
% \begin{tabular}{|c|c|c|c|c|}
% \hline
% \multicolumn{5}{|c|}{\textbf{Traffic Prediction Task}} \\\hline
% \textit{partition} & \textit{energy}  & \textit{time}   &  \textit{loss}  &  \textit{ratio} \\\hline
% 8,1,1,1	 & 0.965	 & 2.343	 & 0.248	 & 7.381	\\\hline
% 64,1,1,1	 & 0.934	 & 2.899	 & 0.251	 & 35.863	\\\hline
% 512,1,1,1	 & 0.938	 & 3.017	 & 0.252	 & 72.126	\\\hline
% 4096,1,1,1	 & 0.904	 & 2.914	 & 0.257	 & 77.928	\\\hline
% 512,512,512,1	 & 1.002	 & 1.257	 & 0.250	 & 53.895	\\\hline
% 4096,4096,4096,1	 & 0.994	 & 1.248	 & 0.251	 & 58.560 \\\hline
% \end{tabular}

% \end{table}

\begin{table}[h]
\centering
\caption{ Energy and time consumed with different dataset partitions.}\label{tab:fairness}
\renewcommand\arraystretch{1}
\begin{tabular}{c|cccc}
\toprule 
\multicolumn{5}{c}{\textbf{Image Task}} \\\midrule
Partition & Energy  & Time   &  Accuracy  &  Ratio \\\midrule
$8:1:1:1$	 & $1.030$	 & $2.287$	 & $0.980$	 & $6.190$	 \\
$64:1:1:1$	 & $1.059	$ & $2.950	$ & $0.980$	 & $19.462$	 \\
$512:1:1:1$	 & $1.058$	 & $3.023$	 & $0.980$	 & $25.548$	 \\
$4096:1:1:1$	 & $1.074	$ & $3.064$	 & $0.980$	 & $25.259$	 \\
$512:512:512:1$	 & $0.996	$ & $1.233$	 & $0.980$	 & $14.031$	 \\
$4096:4096:4096:1$	 & $0.994$	 & $1.235$	 & $0.980$	 & $14.021$	 \\ \bottomrule 
\end{tabular}
\\
\begin{tabular}{c|cccc}
%\toprule
\multicolumn{5}{c}{\textbf{Traffic Task}} \\\midrule
Partition & Energy  & Time   &  Accuracy  &  Ratio \\\midrule
$8:1:1:1$	 & $0.965$	 & $2.343$	 & $0.248$	 & $7.381$	\\
$64:1:1:1$	 & $0.934$	 & $2.899$	 & $0.251$	 & $35.863	$\\
$512:1:1:1$	 & $0.938$	 & $3.017$	 & $0.252$	 & $72.126$	\\
$4096:1:1:1$		 & $0.904$	 & $2.914$	 & $0.257$	 & $77.928$	\\
$512:512:512:1$	 & $1.002$	 & $1.257$	 & $0.250$	 & $53.895$	\\
$4096:4096:4096:1$	 & $0.994$	 & $1.248$	 & $0.251$	 & $58.560$ \\\bottomrule
\end{tabular}

\end{table}

Fairness is also an important metric and should be evaluated when applying an algorithm in multi-agent environments. That is, some agents have more raw data than other agents and thus consume more energy during the local training phase. Such a situation could cause service imbalance and reduced training efficiency.
We simulate this scenario with different dataset partitions and focus on the system consumed energy, time, training performance, and the consumed energy ratio between two agents when the system reaches the termination condition.
In the following experiments, we consider the configurations with one WiFi AP and four agents served by the WiFi AP.
The dataset is partitioned according to the partition ratio at the beginning of each experiment. 
We examine the system outputs when the number of rounds equals 10 for the classification task and 25 for the traffic prediction task. 

%\subsubsection{Effect of Unbalanced Data Distribution}
The partition values given in Table \ref{tab:fairness} denote the dataset size ratio among four agents.
The energy and time in Table \ref{tab:fairness} denotes the total energy and time consumed by agents when reaching the preset termination condition. Also, to make a significant comparison of the energy and time consumed among different partitions, we set the uniform dataset partition as the baseline and present the ratio of energy and time to the baseline.
The ratio presented in Table \ref{tab:fairness} denotes the energy consumed ratio between the agents with the largest and smallest dataset. From Table \ref{tab:fairness}, we can observe that the energy consumed and evaluation results with different partition scenarios stay unchanged within acceptable errors. Although the impact of unbalance datasets is not significant, we can still tell that the unbalanced dataset partitions will affect the training performance for the traffic prediction task. On the other hand, we can observe from Table \ref{tab:fairness} that the unbalanced dataset partitions have significant effects on the total time consumed and energy ratio among different agents for both tasks. Besides, comparing the time columns from two tables, the time consumed ratios for the same partition keep the same within acceptable errors, while this is not the case for the ratios. This is because that the computation to communication consumed energy ratios are different for two tasks, which affects the consumed energy of each agent while does not affect the consumed time.

\subsection{Scalability}
\begin{table}[!t]
\centering

\renewcommand\arraystretch{1.1}
\caption{ Wall-clock time in seconds consumed per round with different numbers of cells and cores.}\label{tab:wallclock}

\begin{tabular}{|c|c|c|c|c|c|c|}
\hline
\multicolumn{7}{|c|}{\textbf{Image Classification Task}} \\\hline
\diagbox{\textit{C}}{\textit{Cr}} & 1  & 2   &  4  &  8 & 16 & 32  \\ \hline
\multirow{2}{1cm}{\centering 32} 	 & 5862	 & 3329 	 & 1595	 & 932	 & 624	 & 623	\\
\cline{2-7}
& 100.0\%  	 & 56.8\%  	 & 27.2\%  	 & 15.9\%  	 & 10.6\%  	 & 10.6\%  \\\hline

\multirow{2}{1cm}{\centering 64} 	 & 12712  	 & 6073  	 & 3318  	 & 1763  	 & 1282  	 & 1404	\\
\cline{2-7}
& 100.0\%  	 & 47.8\%  	 & 26.1\%  	 & 13.9\%  	 & 10.1\%  	 & 11.0\%  \\\hline

\end{tabular}
\\
\begin{tabular}{|c|c|c|c|c|c|c|c|}
\hline
\multicolumn{7}{|c|}{\textbf{Traffic Prediction Task}} \\\hline
\diagbox{\textit{C}}{\textit{Cr}} & 1  & 2   &  4  &  8 & 16 & 32 \\\hline

\multirow{2}{1cm}{\centering 32} 	 & 54  	 & 28  	 & 17  	 & 10  	 & 6  	 & 6 \\\cline{2-7}
 	 & 100.0\%  	 & 52.8\%  	 & 30.2\%  	 & 17.0\%  	 & 11.3\%  	 & 11.3\% \\\hline
\multirow{2}{1cm}{\centering 64} 	 & 110  	 & 58  	 & 32  	 & 19  	 & 13  	 & 14 \\ \cline{2-7}
 	 & 100.0\%  	 & 52.3\%  	 & 28.4\%  	 & 17.4\%  	 & 11.0\%  	 & 11.9\% \\\hline
\multirow{2}{1cm}{\centering 128} 	 & 225  	 & 135  	 & 64  	 & 39  	 & 29  	 & 33 \\\cline{2-7}
 	 & 100.0\%  	 & 59.8\%  	 & 28.6\%  	 & 17.0\%  	 & 12.9\%  	 & 14.3\% \\\hline
\multirow{2}{1cm}{\centering 256} 	 & 496  	 & 247  	 & 133  	 & 88  	 & 80  	 & 93  \\\cline{2-7}
 	 & 100.0\%  	 & 49.7\%  	 & 26.9\%  	 & 17.8\%  	 & 16.2\%  	 & 18.8\% \\\hline

\end{tabular}

\end{table}

Although the scalability of our proposed system is unrelated to the performance of an algorithm, we still would like to emphasize its importance for users when implemented in practice. We evaluate the scalability against the wall-clock (simulation running) time. 
The results for both image classification and traffic tasks are presented in Table \ref{tab:wallclock}. 
The cores and cells in each table header denote the numbers of computing cores and simulated cells utilized in each simulation. The simulated cells are uniformly distributed in all computing cores. 
We present the average running time per round in the first row for each cell scenario. We offer the percent of the wall-clock time of different cores to that with one core in the second row.
Due to the enormous wall-clock consumption of the image classification task, we only conduct the experiments with the number of cells no more than 64.

By comparing the results of different cells within the same number of cores, it is straightforward to observe that the wall-clock increases almost linearly with the number of cells.
By comparing the results of different numbers of cores within the same number of cells, the wall-clock time decreases as expected with the increase in cores. 
However, it takes almost the same wall-clock time to simulate one round when the number of cores equals 16 and 32. This is caused by the limitations on the multi-processing scheme and our hardware platform. Compared to the consumed wall-clock time by the simulator for computing purposes, sharing messages among multiple cores takes more time. As a result, increasing the number of computing cores in this situation will not help to decrease the wall-clock time.

\section{Conclusion}\label{c}
In this paper, we virtualized the basics of \gls{dai} in the wireless environments and proposed the \gls{airdai} system, which is capable of evaluating the training performance metrics and a set of system related QoS metrics. 
In addition,  we introduced a general wireless channel model and analyzed the impacts of operating \gls{dai} under different wireless setups on the convergence rate. The experimental results revealed how wireless transmission parameters and system configurations affect the training efficiency of the \gls{dai} algorithms. 
Based on the proposed \gls{airdai} system, we designed a Python built simulator that works on both single and multiple computing cores and is compatible with existing ML systems.
We took the well-known FedAvg algorithm as an example and conducted extensive experiments by the designed simulator.
The experimental results pertaining to prediction accuracy and QoS metrics verified the effectiveness and efficiency of the proposed system and its associated simulator. 
By this generic system design and the simulator codes provided, the research progress on \gls{dai} in wireless communication systems is expected to be accelerated.

\appendices
\section{Proof of Theorems}

In the appendix, we firstly introduce four general assumptions commonly applied in the SGD convergence analysis. 
Secondly, we introduce the definition of a new term to distinguish the scenarios of iid and non-iid dataset distributions. Then, we introduce the lemmas that gives the limitation of one-step SGD update and linear ratio relationship between learning rates. At last, we give the proof of convergence based on the above two lemmas.

\begin{assumption}
\label{assumption1}
$F_{1}, F_{2}, \cdots, F_{N}$ are all $L$ -smooth: for all $\mathbf{v}$ and $\mathbf{w}$, leading to $F_{k}(\mathbf{v}) \leq F_{k}(\mathbf{w})+(\mathbf{v}-$ $\mathbf{w})^{T} \nabla F_{k}(\mathbf{w})+\frac{L}{2}\|\mathbf{v}-\mathbf{w}\|_{2}^{2}$
\end{assumption}

\begin{assumption}
\label{assumption2}
$F_{1}, F_{2}, \cdots, F_{N}$ are all $\mu$ -strongly convex: for all $\mathbf{v}$ and $\mathbf{w}$, leading to $F_{k}(\mathbf{v}) \geq F_{k}(\mathbf{w})+(\mathbf{v}-$ $\mathbf{w})^{T} \nabla F_{k}(\mathbf{w})+\frac{\mu}{2}\|\mathbf{v}-\mathbf{w}\|_{2}^{2}$
\end{assumption}

\begin{assumption}
\label{assumption3}
Letting $\xi_{t}^{k}$ be randomly sampled from the $k$th device's local data in a uniform manner, the variance of stochastic gradients in each device is bounded by $\mathbb{E}\left\|\nabla F_{k}\left(\mathbf{w}_{t}^{k}, \xi_{t}^{k}\right)-\nabla F_{k}\left(\mathbf{w}_{t}^{k}\right)\right\|^{2} \leq \sigma_{k}^{2}$,
 $\forall~k=1, 2 \cdots, N$
\end{assumption}

\begin{assumption}
\label{assumption4}
The expected squared norm of stochastic gradients is uniformly bounded, i.e.,
$\mathbb{E}\left\|\nabla F_{k}\left(\mathbf{w}_{+}^{k}, \xi_{+}^{k}\right)\right\|^{2} \leq G^{2}$, $\forall~k=1, 2 \cdots, N$ and $\forall~t=0, 1 \cdots, T-1$
for $k=1, \cdots, N$
\end{assumption}
The assumptions mentioned above on functions $F_{1}, F_{2}, \cdots, F_{N}$ are general and necessary for the convergence analysis; 
typical examples include the $\ell_{2}$ -norm regularized linear regression, logistic regression, and softmax classifier.

To extend the analysis on both the iid and non-iid dataset partition scenarios, we propose a new term to quantify the degree of non-iid. The definition is as follows.

\begin{definition}
\label{definition1}
Let $F^{*}$ and $F_{k}^{*}$ be the minimum values of $F$ and $F_{k},$ respectively. We use the term $\Gamma=F^{*}-\sum_{k=1}^{N} p_{k} F_{k}^{*}$ to quantify the degree of heterogeneity of non-iid distributions. That is, if the data are iid, then $\Gamma$ goes to zero as the number of samples grows. If the data are non-iid, then $\Gamma$ is nonzero, and its magnitude signifies the heterogeneity of data distributions.
\end{definition}

With the above assumptions and definition, we formally present Lemma \ref{lemma:one_step}, which limits the expected distance between the current value and the optimum with one-step SGD.

\begin{appd_lemma}
\label{appd_lemma:one_step}
Assume that the central server received $\tilde{N}_{t}$ activated clients in the preset time window.
Letting $\Delta_{t}=\mathbb{E}\left\|\overline{\mathbf{w}}_{t+1}-\mathbf{w}^{\star}\right\|^{2}$, we have
\begin{equation}
    \Delta_{t+1} \leq\left(1-\eta_{t} \mu\right) \Delta_{t}+\eta_{t}^{2} (B + C_t)
\end{equation}
where
$B=\sum_{k=1}^{N} p_{k}^{2} \sigma_{k}^{2}+6 L \Gamma+8(E-1)^{2} G^{2}$
and $C_t = \frac{N-\tilde{N}_{t}}{N-1} \frac{4}{\tilde{N}} E^{2} G^{2}$.
\end{appd_lemma}
\begin{proof}
The proof of the presented lemma can be found in \cite{li2019convergence}. 
\end{proof}

We present Lemma \ref{lemma:linear_ratio} as follows, in which we aim at finding the learning rate relations between the full device participation setting and the partial device participation setting caused due to limited time window.
\begin{appd_lemma}
\label{appd_lemma:linear_ratio}
Denote $\bar{\eta}_t$ to be the learning rate at communication round $t$ to guarantee the algorithm convergence when full devices are participated.
Let $r_t = \frac{\tilde{N}_t}{N}$ be the device participation ratio at communication round $t$. The convergence of the algorithm when partial devices are participated can be guaranteed by setting $\eta_t = r_t\bar{\eta}_t$.
\end{appd_lemma}

\begin{proof}
Let $\bar{\eta}_t = \frac{\mu\tilde{\Delta}_{t}}{2B}$, which implies that $C_t=0$ and the number of clients are all activated, we can obtain the following relations:
\begin{align}
    \eta_t =& \bar{\eta}_t\frac{B}{B+C_t} =\Bar{\eta}_t\left[1+\varepsilon\left(\frac{N-K_t}{K_t}\right)\right]^{-1},
\end{align}
where $\varepsilon = \frac{C_t}{B}\times\frac{\tilde{N}_t}{N-\tilde{N}_t}$ is a $\tilde{N}_t$-irrelevant constant. 
For simplicity, $\varepsilon$ could be stipulated to be unity, and hence, we obtain the following relation
\begin{equation}
\label{lr_adjust}
    \eta_t = \frac{\tilde{N}_t}{N}\Bar{\eta}_t = r_t\Bar{\eta}_t,
\end{equation}
which indicates that we can adapt the learning rate linearly with respect to the number of activated clients. 
\end{proof}

With the lemmas and assumptions mentioned above, we are able to give the bound on the convergence of FedAvg algorithm in wireless environment settings as follows,
\begin{appd_theorem}
Let the assumptions hold and $L, \mu, \sigma_{k}, G$ be defined therein. Choose $\kappa=\frac{L}{\mu}$, $\gamma=\max \{8 \kappa, E\}$ and the learning rate $\eta_{t}=\frac{2r_t}{\mu(\gamma+t)} .$ Then FedAvg algorithm in wireless environments satisfies
$$
\mathbb{E}\left[F\left(\mathbf{w}_{T}\right)\right]-F^{*} \leq \frac{2 \kappa}{\gamma+T}\left(\frac{B+D}{\mu}+2 L\left\|\mathbf{w}_{0}-\mathbf{w}^{*}\right\|^{2}\right),
$$
where $B=\sum_{k=1}^{N} p_{k}^{2} \sigma_{k}^{2}+6 L \Gamma+8(E-1)^{2} G^{2}$, and $D = 4E^2G^2$.
\end{appd_theorem}

\begin{proof}
Our proof starts with the full device participation condition. 
Let $C_t=0$, from Lemma \ref{appd_lemma:one_step} we obtain as follows,
\begin{equation}
    \Delta_{t+1} \leq\left(1-\eta_{t} \mu\right) \Delta_{t}+\eta_{t}^{2} B , 
\end{equation}

For a diminishing step size, $\eta_{t}=\frac{\beta}{t+\gamma}$ for some $\beta>\frac{1}{\mu}$ and $\gamma>0$ such that $\eta_{1} \leq \min \left\{\frac{1}{\mu}, \frac{1}{4 L}\right\}=\frac{1}{4 L}$ and $\eta_{t} \leq 2 \eta_{t+E} .$ We will prove $\Delta_{t} \leq \frac{v}{\gamma+t}$ where $v=\max \left\{\frac{\beta^{2} B}{\beta \mu-1},(\gamma+1) \Delta_{1}\right\}$.
We prove it by induction. Firstly, the definition of $v$ ensures that it holds for $t=1$. Assume the conclusion holds for some $t$, it follows that
\begin{equation}
\begin{aligned}
\Delta_{t+1} & \leq\left(1-\eta_{t} \mu\right) \Delta_{t}+\eta_{t}^{2} B \\
&=\left(1-\frac{\beta \mu}{t+\gamma}\right) \frac{v}{t+\gamma}+\frac{\beta^{2} B}{(t+\gamma)^{2}} \\
&=\frac{t+\gamma-1}{(t+\gamma)^{2}} v+\left[\frac{\beta^{2} B}{(t+\gamma)^{2}}-\frac{\beta \mu-1}{(t+\gamma)^{2}} v\right] \\
& \leq \frac{v}{t+\gamma+1}
\end{aligned}
\end{equation}

Then by the strong convexity of $F(\cdot)$
\begin{equation}
    \mathbb{E}\left[F\left(\overline{\mathbf{w}}_{t}\right)\right]-F^{*} \leq \frac{L}{2} \Delta_{t} \leq \frac{L}{2} \frac{v}{\gamma+t}
\end{equation}

Specifically, if we choose $\beta=\frac{2}{\mu}, \gamma=\max \left\{8 \frac{L}{\mu}-1, E\right\}$ and denote $\kappa=\frac{L}{\mu}$, then $\eta_{t}=\frac{2}{\mu} \frac{1}{\gamma+t}$ and
\begin{equation}
    \mathbb{E}\left[F\left(\overline{\mathbf{w}}_{t}\right)\right]-F^{*} \leq \frac{2 \kappa}{\gamma+t}\left(\frac{B}{\mu}+2 L \Delta_{1}\right)
\end{equation}

For $C_t > 0$ (partial participation), from Lemma \ref{lemma:linear_ratio}, we know that the convergence is guaranteed by setting $\eta_t = r_t\bar{\eta}_t$, where $\bar{\eta}_t$ is the learning rate in full participation condition.
Therefore, let $\eta_t = \frac{2r_t}{\mu(\gamma+t)}$ and replace $B$ with $B+C_t$, we have
\begin{equation}
\begin{aligned}
    \mathbb{E}\left[F\left(\overline{\mathbf{w}}_{t}\right)\right]-F^{*} &\leq \frac{2 \kappa}{\gamma+t}\left(\frac{B+C_t}{\mu}+2 L \Delta_{1}\right)\\
    &\leq 
    \frac{2 \kappa}{\gamma+t}\left(\frac{B+D}{\mu}+2 L \Delta_{1}\right),
\end{aligned}
\end{equation}
where $D=4E^2G^2$ is the upper bound of $C_t$. 
\end{proof}

\ifCLASSOPTIONcaptionsoff
  \newpage
\fi

\bibliographystyle{IEEEtran}
\bibliography{Journal_version}

\begin{thebibliography}{10}
\providecommand{\url}[1]{#1}
\csname url@samestyle\endcsname
\providecommand{\newblock}{\relax}
\providecommand{\bibinfo}[2]{#2}
\providecommand{\BIBentrySTDinterwordspacing}{\spaceskip=0pt\relax}
\providecommand{\BIBentryALTinterwordstretchfactor}{4}
\providecommand{\BIBentryALTinterwordspacing}{\spaceskip=\fontdimen2\font plus
\BIBentryALTinterwordstretchfactor\fontdimen3\font minus
  \fontdimen4\font\relax}
\providecommand{\BIBforeignlanguage}[2]{{%
\expandafter\ifx\csname l@#1\endcsname\relax
\typeout{** WARNING: IEEEtran.bst: No hyphenation pattern has been}%
\typeout{** loaded for the language `#1'. Using the pattern for}%
\typeout{** the default language instead.}%
\else
\language=\csname l@#1\endcsname
\fi
#2}}
\providecommand{\BIBdecl}{\relax}
\BIBdecl

\bibitem{dang2020should}
S.~Dang, O.~Amin, B.~Shihada, and M.-S. Alouini, ``What should {6G} be?''
  \emph{Nature Electronics}, vol.~3, no.~1, pp. 20--29, 2020.

\bibitem{8677314}
M.~Abrams, J.~Abrams, P.~Cullen, and L.~Goldstein, ``Artificial intelligence,
  ethics, and enhanced data stewardship,'' \emph{IEEE Security Privacy},
  vol.~17, no.~2, pp. 17--30, 2019.

\bibitem{8928170}
E.~{Ibarrola}, M.~{Davis}, C.~{Voisin}, C.~{Close}, and L.~{Cristobo}, ``{QoE}
  enhancement in next generation wireless ecosystems: {A} machine learning
  approach,'' \emph{IEEE Communications Standards Magazine}, vol.~3, no.~3, pp.
  63--70, Sept. 2019.

\bibitem{9060868}
W.~Y.~B. Lim, N.~C. Luong, D.~T. Hoang, Y.~Jiao, Y.-C. Liang, Q.~Yang,
  D.~Niyato, and C.~Miao, ``Federated learning in mobile edge networks: {A}
  comprehensive survey,'' \emph{IEEE Communications Surveys Tutorials},
  vol.~22, no.~3, pp. 2031--2063, 2020.

\bibitem{9311931}
M.~Chen, H.~V. Poor, W.~Saad, and S.~Cui, ``Wireless communications for
  collaborative federated learning,'' \emph{IEEE Communications Magazine},
  vol.~58, no.~12, pp. 48--54, 2020.

\bibitem{6573299}
M.~{Lesk}, ``Big data, big brother, big money,'' \emph{IEEE Security Privacy},
  vol.~11, no.~4, pp. 85--89, July 2013.

\bibitem{8770530}
X.~{Wang}, Y.~{Han}, C.~{Wang}, Q.~{Zhao}, X.~{Chen}, and M.~{Chen}, ``In-edge
  {AI}: {Intelligentizing} mobile edge computing, caching and communication by
  federated learning,'' \emph{IEEE Network}, vol.~33, no.~5, pp. 156--165,
  Sept. 2019.

\bibitem{dean2012large}
J.~Dean, G.~S. Corrado, R.~Monga, K.~Chen, M.~Devin, Q.~V. Le, M.~Z. Mao,
  M.~Ranzato, A.~Senior, P.~Tucker \emph{et~al.}, ``Large scale distributed
  deep networks,'' in \emph{Proceedings of the 25th International Conference on
  Neural Information Processing Systems-Volume 1}, 2012, pp. 1223--1231.

\bibitem{bottou2012stochastic}
L.~Bottou, ``Stochastic gradient descent tricks,'' in \emph{Neural networks:
  Tricks of the trade}.\hskip 1em plus 0.5em minus 0.4em\relax Springer, 2012,
  pp. 421--436.

\bibitem{stich2018local}
S.~U. Stich, ``Local {SGD} converges fast and communicates little,'' in
  \emph{ICLR 2019-International Conference on Learning Representations}, 2019,
  pp. 1--17.

\bibitem{konevcny2016federated}
J.~Kone{\v{c}}n{\`y}, H.~B. McMahan, F.~X. Yu, P.~Richt{\'a}rik, A.~T. Suresh,
  and D.~Bacon, ``Federated learning: {Strategies} for improving communication
  efficiency,'' \emph{arXiv preprint arXiv:1610.05492}, 2016.

\bibitem{bond2014readings}
A.~H. Bond and L.~Gasser, \emph{Readings in distributed artificial
  intelligence}.\hskip 1em plus 0.5em minus 0.4em\relax Morgan Kaufmann, 2014.

\bibitem{9226618}
J.~Zhou, S.~Dang, B.~Shihada, and M.-S. Alouini, ``Power allocation for relayed
  {OFDM} with index modulation assisted by artificial neural network,''
  \emph{IEEE Wireless Communications Letters}, vol.~10, no.~2, pp. 373--377,
  2021.

\bibitem{9058719}
S.~Dang, M.~Wen, S.~Mumtaz, J.~Li, and C.~Li, ``Enabling multi-carrier relay
  selection by sensing fusion and cascaded {ANN} for intelligent vehicular
  communications,'' \emph{IEEE Sensors Journal}, pp. 1--1, 2020.

\bibitem{9365714}
J.~Li, S.~Dang, M.~Wen, Z.~Zhang, and Q.~Li, ``Smart detection using the
  cascaded artificial neural network for {OFDM} with subcarrier number
  modulation,'' \emph{IEEE Wireless Communications Letters}, pp. 1--1, 2021.

\bibitem{9084352}
T.~Li, A.~K. Sahu, A.~Talwalkar, and V.~Smith, ``Federated learning:
  {Challenges}, methods, and future directions,'' \emph{IEEE Signal Processing
  Magazine}, vol.~37, no.~3, pp. 50--60, 2020.

\bibitem{mcmahan2017communication}
B.~McMahan, E.~Moore, D.~Ramage, S.~Hampson, and B.~A. y~Arcas,
  ``Communication-efficient learning of deep networks from decentralized
  data,'' in \emph{Artificial Intelligence and Statistics}.\hskip 1em plus
  0.5em minus 0.4em\relax PMLR, 2017, pp. 1273--1282.

\bibitem{8994206}
J.~Kang, Z.~Xiong, D.~Niyato, Y.~Zou, Y.~Zhang, and M.~Guizani, ``Reliable
  federated learning for mobile networks,'' \emph{IEEE Wireless
  Communications}, vol.~27, no.~2, pp. 72--80, 2020.

\bibitem{9311906}
S.~Hosseinalipour, C.~G. Brinton, V.~Aggarwal, H.~Dai, and M.~Chiang, ``From
  federated to fog learning: Distributed machine learning over heterogeneous
  wireless networks,'' \emph{IEEE Communications Magazine}, vol.~58, no.~12,
  pp. 41--47, 2020.

\bibitem{9170265}
Y.~Liu, J.~Peng, J.~Kang, A.~M. Iliyasu, D.~Niyato, and A.~A.~A. El-Latif, ``A
  secure federated learning framework for {5G} networks,'' \emph{IEEE Wireless
  Communications}, vol.~27, no.~4, pp. 24--31, 2020.

\bibitem{9141214}
S.~Niknam, H.~S. Dhillon, and J.~H. Reed, ``Federated learning for wireless
  communications: Motivation, opportunities, and challenges,'' \emph{IEEE
  Communications Magazine}, vol.~58, no.~6, pp. 46--51, 2020.

\bibitem{9210812}
M.~Chen, Z.~Yang, W.~Saad, C.~Yin, H.~V. Poor, and S.~Cui, ``A joint learning
  and communications framework for federated learning over wireless networks,''
  \emph{IEEE Transactions on Wireless Communications}, vol.~20, no.~1, pp.
  269--283, 2021.

\bibitem{bonawitz2019towards}
K.~Bonawitz, H.~Eichner, W.~Grieskamp, D.~Huba, A.~Ingerman, V.~Ivanov,
  C.~Kiddon, J.~Konecny, S.~Mazzocchi, H.~B. McMahan \emph{et~al.}, ``Towards
  federated learning at scale: {System} design,'' \emph{arXiv preprint
  arXiv:1902.01046}, 2019.

\bibitem{he2020fedml}
C.~He, S.~Li, J.~So, X.~Zeng, M.~Zhang, H.~Wang, X.~Wang, P.~Vepakomma,
  A.~Singh, H.~Qiu \emph{et~al.}, ``Fedml: A research library and benchmark for
  federated machine learning,'' \emph{arXiv preprint arXiv:2007.13518}, 2020.

\bibitem{li2019federated}
Q.~Li, Z.~Wen, and B.~He, ``Federated learning systems: Vision, hype and
  reality for data privacy and protection.'' 2019.

\bibitem{paszke2019pytorch}
A.~Paszke, S.~Gross, F.~Massa, A.~Lerer, J.~Bradbury, G.~Chanan, T.~Killeen,
  Z.~Lin, N.~Gimelshein, L.~Antiga \emph{et~al.}, ``Pytorch: An imperative
  style, high-performance deep learning library,'' in \emph{Advances in Neural
  Information Processing Systems}, 2019, pp. 8024--8035.

\bibitem{abadi2016tensorflow}
M.~Abadi, P.~Barham, J.~Chen, Z.~Chen, A.~Davis, J.~Dean, M.~Devin,
  S.~Ghemawat, G.~Irving, M.~Isard \emph{et~al.}, ``Tensorflow: A system for
  large-scale machine learning,'' in \emph{12th {USENIX} symposium on operating
  systems design and implementation ({OSDI} 16)}, 2016, pp. 265--283.

\bibitem{kairouz2019advances}
P.~Kairouz, H.~B. McMahan, B.~Avent, A.~Bellet, M.~Bennis, A.~N. Bhagoji,
  K.~Bonawitz, Z.~Charles, G.~Cormode, R.~Cummings \emph{et~al.}, ``Advances
  and open problems in federated learning,'' \emph{arXiv preprint
  arXiv:1912.04977}, 2019.

\bibitem{han2020adaptive}
P.~Han, S.~Wang, and K.~K. Leung, ``Adaptive gradient sparsification for
  efficient federated learning: An online learning approach,'' in \emph{2020
  IEEE 40th International Conference on Distributed Computing Systems
  (ICDCS)}.\hskip 1em plus 0.5em minus 0.4em\relax IEEE, 2020, pp. 300--310.

\bibitem{wangni2017gradient}
J.~Wangni, J.~Wang, J.~Liu, and T.~Zhang, ``Gradient sparsification for
  communication-efficient distributed optimization,'' \emph{arXiv preprint
  arXiv:1710.09854}, 2017.

\bibitem{alistarh2017qsgd}
D.~Alistarh, D.~Grubic, J.~Li, R.~Tomioka, and M.~Vojnovic, ``Qsgd:
  Communication-efficient sgd via gradient quantization and encoding,''
  \emph{Advances in Neural Information Processing Systems}, vol.~30, pp.
  1709--1720, 2017.

\bibitem{8737464}
N.~H. {Tran}, W.~{Bao}, A.~{Zomaya}, M.~N.~H. {Nguyen}, and C.~S. {Hong},
  ``Federated learning over wireless networks: {Optimization} model design and
  analysis,'' in \emph{Proc. IEEE INFOCOM}, Paris, France, Apr. 2019, pp.
  1387--1395.

\bibitem{dinh2019federated}
C.~T. Dinh, N.~H. Tran, M.~N.~H. Nguyen, C.~S. Hong, W.~Bao, A.~Y. Zomaya, and
  V.~Gramoli, ``Federated learning over wireless networks: Convergence analysis
  and resource allocation,'' \emph{IEEE/ACM Transactions on Networking},
  vol.~29, no.~1, pp. 398--409, 2021.

\bibitem{chen2020convergence}
M.~Chen, H.~V. Poor, W.~Saad, and S.~Cui, ``Convergence time optimization for
  federated learning over wireless networks,'' \emph{arXiv preprint
  arXiv:2001.07845}, 2020.

\bibitem{8950073}
S.~Savazzi, M.~Nicoli, and V.~Rampa, ``Federated learning with cooperating
  devices: {A} consensus approach for massive {IoT} networks,'' \emph{IEEE
  Internet of Things Journal}, vol.~7, no.~5, pp. 4641--4654, 2020.

\bibitem{8963610}
Y.~{Zhan}, P.~{Li}, Z.~{Qu}, D.~{Zeng}, and S.~{Guo}, ``A learning-based
  incentive mechanism for federated learning,'' \emph{IEEE Internet of Things
  Journal}, 2020.

\bibitem{8917592}
S.~{Samarakoon}, M.~{Bennis}, W.~{Saad}, and M.~{Debbah}, ``Distributed
  federated learning for ultra-reliable low-latency vehicular communications,''
  \emph{IEEE Transactions on Communications}, vol.~68, no.~2, pp. 1146--1159,
  Feb. 2020.

\bibitem{tran2019federated}
N.~H. Tran, W.~Bao, A.~Zomaya, N.~M. NH, and C.~S. Hong, ``Federated learning
  over wireless networks: Optimization model design and analysis,'' in
  \emph{IEEE INFOCOM 2019-IEEE Conference on Computer Communications}.\hskip
  1em plus 0.5em minus 0.4em\relax IEEE, 2019, pp. 1387--1395.

\bibitem{sprague2018asynchronous}
M.~R. Sprague, A.~Jalalirad, M.~Scavuzzo, C.~Capota, M.~Neun, L.~Do, and
  M.~Kopp, ``Asynchronous federated learning for geospatial applications,'' in
  \emph{Joint European Conference on Machine Learning and Knowledge Discovery
  in Databases}.\hskip 1em plus 0.5em minus 0.4em\relax Springer, 2018, pp.
  21--28.

\bibitem{nsdiscrete}
A.~NS, ``Discrete-event network simulator for internet systems.''

\bibitem{s2010wireless}
T.~S. Rappaport, \emph{Wireless Communications: Principles and Practice}.\hskip
  1em plus 0.5em minus 0.4em\relax Pearson Education, 2010.

\bibitem{rubio2007evaluation}
L.~Rubio, J.~Reig, and N.~Cardona, ``Evaluation of {Nakagami} fading behaviour
  based on measurements in urban scenarios,'' \emph{AEU-International Journal
  of Electronics and Communications}, vol.~61, no.~2, pp. 135--138, 2007.

\bibitem{barsocchi2006channel}
P.~Barsocchi, ``Channel models for terrestrial wireless communications: {A}
  survey,'' \emph{CNR-ISTI Technical Report}, vol.~83, 2006.

\bibitem{proakis2008digital}
J.~Proakis and M.~Salehi, \emph{Digital Communications}.\hskip 1em plus 0.5em
  minus 0.4em\relax McGraw-Hill, 2008.

\bibitem{wang2018cooperative}
J.~Wang and G.~Joshi, ``Cooperative {SGD}: {A} unified framework for the design
  and analysis of communication-efficient {SGD} algorithms,'' in \emph{ICML
  Workshop on Coding Theory for Machine Learning}, 2019.

\bibitem{li2019convergence}
X.~Li, K.~Huang, W.~Yang, S.~Wang, and Z.~Zhang, ``On the convergence of
  {FedAvg} on non-iid data,'' \emph{arXiv preprint arXiv:1907.02189}, 2019.

\bibitem{zhangcl2018}
C.~Zhang, H.~Zhang, D.~Yuan, and M.~Zhang, ``Citywide cellular traffic
  prediction based on densely connected convolutional neural networks,''
  \emph{IEEE Communications Letters}, vol.~22, no.~8, pp. 1656--1659, 2018.

\bibitem{zhangjsac2019}
C.~Zhang, H.~Zhang, J.~Qiao, D.~Yuan, and M.~Zhang, ``Deep transfer learning
  for intelligent cellular traffic prediction based on cross-domain big data,''
  \emph{IEEE Journal on Selected Areas in Communications}, vol.~37, no.~6, pp.
  1389--1401, 2019.

\bibitem{zhanginfocom2021}
C.~Zhang, S.~Dang, B.~Shihada, and M.-S. Alouini, ``Dual attention-based
  federated learning for wireless traffic prediction,'' \emph{IEEE
  International Conference on Computer Communications (INFOCOM)}, 2021.

\bibitem{Lecun_mnist}
Y.~{Lecun}, L.~{Bottou}, Y.~{Bengio}, and P.~{Haffner}, ``Gradient-based
  learning applied to document recognition,'' \emph{Proceedings of the IEEE},
  vol.~86, no.~11, pp. 2278--2324, Nov 1998.

\bibitem{barlacchi2015multi}
G.~Barlacchi, M.~De~Nadai, R.~Larcher, A.~Casella, C.~Chitic, G.~Torrisi,
  F.~Antonelli, A.~Vespignani, A.~Pentland, and B.~Lepri, ``A multi-source
  dataset of urban life in the city of milan and the province of {Trentino},''
  \emph{Scientific data}, vol.~2, p. 150055, 2015.

\bibitem{carroll2010analysis}
A.~Carroll, G.~Heiser \emph{et~al.}, ``An analysis of power consumption in a
  smartphone.'' in \emph{USENIX annual technical conference}, vol.~14.\hskip
  1em plus 0.5em minus 0.4em\relax Boston, MA, 2010, pp. 21--21.

\bibitem{bhagoji2019analyzing}
A.~N. Bhagoji, S.~Chakraborty, P.~Mittal, and S.~Calo, ``Analyzing federated
  learning through an adversarial lens,'' in \emph{International Conference on
  Machine Learning}.\hskip 1em plus 0.5em minus 0.4em\relax PMLR, 2019, pp.
  634--643.

\bibitem{bagdasaryan2020backdoor}
E.~Bagdasaryan, A.~Veit, Y.~Hua, D.~Estrin, and V.~Shmatikov, ``How to backdoor
  federated learning,'' in \emph{International Conference on Artificial
  Intelligence and Statistics}.\hskip 1em plus 0.5em minus 0.4em\relax PMLR,
  2020, pp. 2938--2948.

\end{thebibliography}

% \end{IEEEbiography}

\end{document}